\documentclass[prodmode,acmjacm]{acmsmall}

\usepackage[ruled]{algorithm2e}
\renewcommand{\algorithmcfname}{ALGORITHM}
\SetAlFnt{\small}
\SetAlCapFnt{\small}

\usepackage{stmaryrd}
\usepackage{amsfonts}
\usepackage{enumerate}
\usepackage{amsmath}
\usepackage{amssymb}
\usepackage{mathrsfs}
\usepackage{url}
\usepackage{graphicx}

\providecommand{\fin}{\color{black}}

\newtheorem{problem}[theorem]{Problem}
\acmVolume{0}
\acmNumber{0}
\acmArticle{0}
\acmYear{2013}
\acmMonth{6}

\begin{document}

\markboth{C.~J.~Hillar and L.-H.~Lim}{Most Tensor Problems are NP-Hard}

\title{Most Tensor Problems are NP-Hard}
\author{CHRISTOPHER~J.~HILLAR
\affil{Mathematical Sciences Research Institute}
LEK-HENG~LIM
\affil{University of Chicago}}

\begin{abstract}
We prove that multilinear (tensor) analogues of many efficiently computable problems in numerical linear algebra are NP-hard. 
Our list here includes: determining the feasibility of a system of bilinear equations, deciding whether a $3$-tensor possesses a given  eigenvalue, singular value, or  spectral norm; approximating an eigenvalue, eigenvector, singular vector, or the spectral norm; and determining the rank or best rank-$1$ approximation of a $3$-tensor. 
Furthermore, we show that restricting these problems to symmetric tensors does not alleviate their NP-hardness.  We also explain how deciding nonnegative definiteness of a symmetric $4$-tensor is NP-hard and how computing the combinatorial hyperdeterminant of a $4$-tensor is NP-, \#P-, and VNP-hard.  We shall argue that our results provide another view of the boundary separating the computational tractability of linear/convex problems from the intractability of nonlinear/nonconvex ones.
\end{abstract}




\acmformat{Hillar, C.~J., and Lim, L.-H.  2012. Most tensor problems are NP-hard}

\begin{bottomstuff}
Hillar was partially supported by an NSA Young Investigators Grant and an NSF All-Institutes Postdoctoral Fellowship administered by the Mathematical Sciences Research Institute through its core grant DMS-0441170. Lim was partially supported by an NSF CAREER Award DMS-1057064, an NSF Collaborative Research Grant DMS 1209136, and an AFOSR Young Investigator Award FA9550-13-1-0133.

Author's addresses: C.~J.~Hillar, Mathematical Sciences Research Institute, Berkeley, CA 94720, \url{chillar@msri.org}; L.-H.~Lim, Computational and Applied Mathematics Initiative,
Department of Statistics, University of Chicago, Chicago, IL 60637, \url{lekheng@galton.uchicago.edu}.
\end{bottomstuff}
\maketitle

\section{Introduction}

Frequently a problem in science or engineering can be reduced to solving a linear (matrix) system of equations and inequalities.  Other times, solutions involve the extraction of certain quantities from matrices such as eigenvectors or singular values.   In computer vision, for instance, segmentations of a digital picture along object boundaries can be found by computing the top eigenvectors of a certain matrix produced from the image \cite{ShiMalik}. 
Another common problem formulation is to find low-rank matrix approximations that explain a given two-dimensional array of data, accomplished, as is now standard, by zeroing the smallest singular values in a singular value decomposition of the array \cite{SVD1,SVD2}.  In general, efficient and reliable routines computing answers to these and similar problems have been a workhorse for real-world applications of computation.

Recently, there has been a flurry of work on multilinear analogues to the basic problems of linear algebra. These ``tensor methods" have found applications in many fields, including approximation algorithms \cite{DKKV,BV}, computational biology \cite{Cart}, computer graphics \cite{VT3}, computer vision \cite{NTF,VT2}, data analysis \cite{CB}, graph theory \cite{F,FW}, neuroimaging \cite{SchultzS}, pattern recognition \cite{VT1}, phylogenetics \cite{AllmanR}, quantum computing \cite{miyake1}, scientific computing \cite{BM}, signal processing \cite{Como94:SP,Como04:ieeesp,KR1}, spectroscopy \cite{SBG}, and wireless communication \cite{SidiBG00:ieeesp}, among other areas. Thus, tensor generalizations to the standard algorithms of linear algebra have the potential to substantially enlarge the arsenal of core tools in numerical computation. 

The main results of this paper, however, support the view that tensor problems are almost invariably computationally hard.  Indeed, we shall prove that many naturally occurring problems for $3$-tensors are NP-hard; that is, solutions to the hardest problems in NP can be found by answering questions about 3-tensors.  A full list of the problems we study can be found in Table~\ref{tab:tensor_complexity} below.  Since we deal with mathematical questions over fields (such as the real numbers $\mathbb{R}$), algorithmic complexity is a somewhat subtle notion.  Our perspective here will be the Turing model of computation \cite{Turing} and the Cook--Karp--Levin model of complexity involving NP-hard \cite{Kn1,Kn2} and NP-complete problems \cite{Cook,Karp,Levin}, as opposed to other computational models \cite{Valiant,BSS,weihrauch2000}.  We describe our framework in Subsection \ref{comp_complexity}  along with a comparison to other models.

\renewcommand\arraystretch{1.5}
\begin{table}%
\tbl{Tractability of Tensor Problems\label{tab:tensor_complexity}}{%
\begin{tabular}{l|l}
\textbf{Problem} & \textbf{Complexity}  \\\hline 
Bivariate Matrix Functions over $\mathbb{R}$, $\mathbb{C}$ & Undecidable (Proposition~\ref{prop:biv}) \\  \hline
Bilinear System over $\mathbb{R}$, $\mathbb{C}$ & NP-hard (Theorems~\ref{QFHard},~\ref{zerosingvaluethm},~\ref{tri_quad_feas_thm}) \\  \hline
Eigenvalue over $\mathbb{R}$  & NP-hard (Theorem~\ref{TEvalueNP}) \\\hline 
Approximating Eigenvector over $\mathbb{R}$ & NP-hard (Theorem~\ref{approx_evect}) \\\hline
Symmetric Eigenvalue  over $\mathbb{R}$ & NP-hard (Theorem~\ref{symm_evalue_nphard_thm})\\  \hline  
Approximating Symmetric Eigenvalue over $\mathbb{R}$ & NP-hard (Theorem~\ref{cor:npzpp}) \\\hline 
Singular Value over $\mathbb{R}$, $\mathbb{C}$ & NP-hard (Theorem~\ref{sing_nphard_thm}) \\\hline 
Symmetric Singular Value over $\mathbb{R}$ & NP-hard (Theorem~\ref{thm:symmresults}) \\\hline 
Approximating Singular Vector over $\mathbb{R}$, $\mathbb{C}$  & NP-hard (Theorem~\ref{approx_singvect}) \\\hline 
Spectral Norm over $\mathbb R$ & NP-hard (Theorem~\ref{spec_norm_nphard_thm}) \\\hline 
Symmetric Spectral Norm  over $\mathbb R$ & NP-hard (Theorem~\ref{thm:symmresults}) \\\hline 
Approximating Spectral Norm  over $\mathbb R$ & NP-hard (Theorem~\ref{approx_spec_nphard_thm})  \\ \hline   
Nonnegative Definiteness & NP-hard  (Theorem~\ref{thm:nonneg}) \\ \hline
Best Rank-$1$ Approximation & NP-hard (Theorem~\ref{rank1approx_nphard_thm}) \\ \hline
Best Symmetric Rank-$1$ Approximation & NP-hard (Theorem~\ref{thm:symmresults}) \\ \hline
Rank over $\mathbb{R}$ or $\mathbb{C}$ & NP-hard (Theorem~\ref{rank_nphard_thm})  \\ \hline  
Enumerating Eigenvectors over $\mathbb{R}$ & \#P-hard (Corollary~\ref{sharp_p_3col}) \\\hline
Combinatorial Hyperdeterminant & NP-, \#P-, VNP-hard (Theorems~\ref{thm:bar} , \ref{thm:gur}, Corollary~\ref{cor:gur})  \\ \hline\hline
Geometric  Hyperdeterminant  & Conjectures~\ref{hyp_conj},~\ref{hyp_approx_conjec} \\  \hline    
Symmetric Rank  & Conjecture~\ref{sym_rank_conj} \\  \hline
Bilinear Programming  & Conjecture~\ref{prob:BP} \\  \hline
Bilinear Least Squares  & Conjecture~\ref{prob:BLS} \\  \hline
\end{tabular}}
\begin{tabnote}%
\Note{Note:}{Except for positive definiteness and the combinatorial hyperdeterminant, which apply to $4$-tensors, all problems refer to the $3$-tensor case.}
\end{tabnote}
\end{table}

One way to interpret these findings is that $3$-tensor problems form a boundary separating classes of tractable linear/convex problems from intractable nonlinear/nonconvex ones. More specifically, linear algebra is concerned with (inverting) vector-valued functions that are locally of the form $f(\mathbf{x}) = \mathbf{b} + A\mathbf{x}$; while convex analysis deals with (minimizing) scalar-valued functions that are locally approximated by $f(\mathbf{x}) = c + \mathbf{b}^\top \mathbf{x} + \mathbf{x}^\top A\mathbf{x}$ with $A$ positive definite.  These functions involve tensors of order $0$, $1$, and $2$: $c \in \mathbb{R}$, $\mathbf{b} \in \mathbb{R}^n$, and $A \in \mathbb{R}^{n \times n}$.  However, as soon as we move on to bilinear vector-valued or trilinear real-valued functions, we invariably come upon $3$-tensors $\mathcal{A} \in \mathbb{R}^{n \times n \times n}$ and the NP-hardness associated with inferring properties of them.

The primary audience for this article are numerical analysts and computational algebraists, although we hope it will be of interest to users of tensor methods in various communities. Parts of our exposition contain standard material (e.g., complexity theory to computer scientists, hyperdeterminants to algebraic geometers, KKT conditions to optimization theorists, etc.), but to appeal to the widest possible audience at the intersection of computer science, linear and multilinear algebra, algebraic geometry, numerical analysis, and optimization, we have keep our discussion as self-contained as possible.  A side contribution is a useful framework for incorporating features of computation over $\mathbb R$ and $\mathbb C$ with classical tools and models of algorithmic complexity involving Turing machines that we think is unlike any existing treatments \cite{BSS,BCSS,HochbaumShanti,Vavasis}.

\subsection{Tensors}

We begin by first defining our basic mathematical objects.   Fix a field $\mathbb{F}$, which for us will be either the rationals $\mathbb{Q}$, the reals $\mathbb{R}$, or the complex numbers $\mathbb{C}$. Also, let $l$, $m$, and $n$ be positive integers.  For the purposes of this article, a $3$-\textit{tensor} $\mathcal{A}$ over $\mathbb{F}$ is an $l\times m\times n$ array of elements of $\mathbb{F}$:
\begin{equation}\label{eq:tensor}
\mathcal{A}=\llbracket a_{ijk}\rrbracket_{i,j,k=1}^{l,m,n}\in\mathbb{F}
^{l\times m\times n}.
\end{equation}
These objects are natural multilinear generalizations of matrices in the following way.

For any positive integer $d$, let $\mathbf{e}_1,\dots,\mathbf{e}_d$ denote the standard basis\footnote{Formally, $\mathbf{e}_i$ is the vector in $\mathbb{F}^d$ with a $1$ in the $i$th coordinate and zeroes everywhere else. In this article, vectors in $\mathbb{F}^n$ will always be column-vectors.} in the $\mathbb{F}$-vector space $\mathbb{F}^d$.  A bilinear function $f: \mathbb{F}^m \times \mathbb{F}^n \to \mathbb{F}$ can be encoded by a matrix $A =[a_{ij}]_{i,j=1}^{m,n} \in \mathbb{F}^{m \times n}$, in which the entry $a_{ij}$ records the value of $f(\mathbf{e}_i,\mathbf{e}_j) \in \mathbb{F}$.  By linearity in each coordinate, specifying $A$ determines the values of $f$ on all of $\mathbb{F}^m \times \mathbb{F}^n$; in fact, we have $f(\mathbf{u},\mathbf{v}) = \mathbf{u}^{\top}A\mathbf{v}$ for any vectors $\mathbf{u} \in \mathbb{F}^m$ and $\mathbf{v}  \in \mathbb{F}^n$.  Thus,  matrices both encode $2$-dimensional arrays of numbers and specify all bilinear functions.  Notice also that if $m=n$ and $A = A^{\top}$ is symmetric, then \[f(\mathbf{u},\mathbf{v}) = \mathbf{u}^{\top}A\mathbf{v} = (\mathbf{u}^{\top}A\mathbf{v})^{\top} = \mathbf{v}^{\top}A^{\top}\mathbf{u} = \mathbf{v}^{\top}A\mathbf{u}  = f(\mathbf{v},\mathbf{u}).\]  Thus, symmetric matrices are bilinear maps invariant under coordinate exchange.

These notions generalize: a $3$-tensor is a trilinear function $f: \mathbb{F}^l \times \mathbb{F}^m \times \mathbb{F}^n \to \mathbb{F}$ which has a coordinate representation given by a \textit{hypermatrix}\footnote{We will not use the term hypermatrix but will simply regard a tensor as synonymous with its coordinate representation. See \cite{hla} for more details.} $\mathcal{A}$ as in \eqref{eq:tensor}.  The subscripts and superscripts in \eqref{eq:tensor} will be dropped whenever the range of $i,j,k$ is obvious or unimportant.
Also, a $3$-tensor $\llbracket a_{ijk}\rrbracket_{i,j,k=1}^{n,n,n}\in\mathbb{F}
^{n\times n\times n}$ is \textit{symmetric} if
\begin{equation}
a_{ijk} = a_{ikj} = a_{jik} = a_{jki} = a_{kij} = a_{kji}. \label{eq:symm}
\end{equation}
These  are coordinate representations of trilinear maps $f: \mathbb{F}^n \times \mathbb{F}^n \times \mathbb{F}^n \to \mathbb{F}$ with
\[
f(\mathbf{u},\mathbf{v},\mathbf{w})=f(\mathbf{u},\mathbf{w},\mathbf{v})=f(\mathbf{v},\mathbf{u},\mathbf{w})=f(\mathbf{v},\mathbf{w},\mathbf{u})=f(\mathbf{w},\mathbf{u},\mathbf{v})=f(\mathbf{w},\mathbf{v},\mathbf{u}).
\]
We focus here on $3$-tensors mainly for expositional purposes.  One exception is the problem of deciding positive definiteness of a tensor, a notion  nontrivial only in even orders.

When $\mathbb{F} = \mathbb{C}$, one may argue that a generalization of the notion of Hermitian or self-adjoint matrices would be more appropriate than that of symmetric matrices. For $3$-tensors, such ``self-adjointness"  depends on a choice of a trilinear form, which might be natural in certain applications \cite{WG}. Our complexity results for symmetric tensors apply as long as the chosen notion reduces to \eqref{eq:symm} for $\mathcal{A} \in \mathbb{R}^{n\times n\times n}$.   

\subsection{Tensor Eigenvalue}

We now explain in detail the tensor eigenvalue problem since it is the simplest multilinear generalization. We shall also use the problem  to illustrate many of the concepts that arise when studying other, more difficult, tensor problems.  The basic notions for eigenvalues of tensors were introduced independently in \cite{L2} and \cite{Qi}, with more developments appearing in \cite{NiQi,Qi2}.  Additional theory from the perspective of toric algebraic geometry and intersection theory was provided recently in \cite{CartSturm}.  We will describe the ideas more formally in Section~\ref{sec:EVP}, but for now it suffices to say that the usual eigenvalues and eigenvectors of a matrix $A\in\mathbb{R}^{n\times n}$ are the stationary values and  points of its Rayleigh quotient, and this  view generalizes to higher order tensors. This gives, for example, an \textit{eigenvector} of a tensor $\mathcal{A} = \llbracket a_{ijk}\rrbracket_{i,j,k=1}^{n,n,n}\in\mathbb{F}
^{n\times n\times n}$ as a nonzero vector $\mathbf{x}  = [x_1,\dots,x_n]^{\top} \in \mathbb{F}^n$ satisfying:
\begin{equation}\label{l2eig}
\sum_{i,j=1}^{n}a_{ijk}x_{i}x_{j}=\lambda x_{k},\quad k=1,\dots,n,
\end{equation}
for some $\lambda \in \mathbb{F}$, which is called an \textit{eigenvalue} of $\mathcal{A}$.  Notice that if $(\lambda,\mathbf{x})$ is an eigenpair, then so is $(t\lambda,t\mathbf{x})$ for any $t \neq 0$; thus, eigenpairs are more naturally defined projectively.

As in the matrix case, generic or ``random" tensors over $\mathbb{F} = \mathbb{C}$ have a finite number of eigenvalues and eigenvectors (up to this scaling equivalence), although their count is exponential in $n$.  Still, it is possible for a tensor to have an infinite number of non-equivalent eigenvalues, but in that case they  comprise a cofinite set of complex numbers.  Another important fact is that over the reals ($\mathbb{F} = \mathbb{R}$), every $3$-tensor has a real eigenpair.  These results and more can be found in \cite{CartSturm}.   The following problem is natural for applications.

\begin{problem}\label{TensorEvalue}
Given $\mathcal{A}\in\mathbb{F}^{n\times n\times n}$, find $(\lambda,\mathbf{x}) \in \mathbb{F} \times \mathbb{F}^n$ with $\mathbf{x} \neq \mathbf{0}$ satisfying \eqref{l2eig}.
\end{problem}


We first discuss the computability of this problem.  When the entries of the tensor $\mathcal{A}$ are real numbers, there is an effective procedure that will output a finite presentation of all real eigenpairs.  A good reference for such methods in real algebraic geometry is \cite{bochnak1998}, and an overview of recent intersections between mathematical logic and algebraic geometry, more generally, can be found in \cite{haskell2000}.  Over $\mathbb C$, this problem can be tackled directly by computing a Gr\"obner basis with Buchberger's algorithm \cite{buchberger1970} since an eigenpair is a solution to a system of polynomial equations over an algebraically closed field (e.g.\ \cite[Example 3.5]{CartSturm}).  Another approach is to work with Macaulay matrices of multivariate resultants \cite{NiQi}.   References for such techniques suitable for numerical analysts are \cite{CLO1,CLO2}.

%

Even though solutions to tensor problems are computable, all known methods quickly become impractical as the tensors become larger (i.e., as $n$ grows).  In principle, this occurs because simply listing the output to Problem~\ref{TensorEvalue} is already prohibitive. It is natural, therefore, to ask for faster methods  checking whether a given $\lambda$ is an eigenvalue or approximating a single eigenpair.  We first analyze the following easier decision problem.

\begin{problem}[Tensor $\lambda$-eigenvalue]\label{TensorZeroEvalue}
Let $\mathbb{F} = \mathbb{R}$ or $\mathbb{C}$, and fix $\lambda \in \mathbb{Q}$.  Decide if $\lambda$ is an eigenvalue (with corresponding eigenvector in $\mathbb{F}^n$) of a tensor $\mathcal{A} \in \mathbb{Q}^{n \times n \times n}$.
\end{problem}


Before explaining our results on Problem \ref{TensorZeroEvalue} and other tensor questions, we define the model of computational complexity that we shall utilize to study them.

\subsection{Computability and Complexity}\label{comp_complexity}
We hope this article will be useful to casual users of computational complexity --- such as numerical analysts and optimization theorists --- who nevertheless desire to understand the tractability of their problems in light of modern complexity theory.  This section and the next provide a high-level overview for such an audience.  In addition, we also carve out a perspective for real computation within the Turing machine framework that we feel is easier to work with than those proposed in \cite{BSS,BCSS,HochbaumShanti,Vavasis}. For readers who have no particular interest in tensor problems, the remainder of our article may then be viewed as a series of instructive examples showing how one may deduce the tractability of a numerical computing problem using the rich collection of NP-complete combinatorial problems.

Computational complexity is usually specified on the following three levels. 

\begin{enumerate}[\upshape I.]
\item \textit{Model of Computation}: What are  inputs and outputs?  What is a computation?  For us, inputs will be rational numbers and outputs will be rational vectors or \textsc{yes}/\textsc{no} responses, and computations are performed on a \textbf{Turing Machine}. Alternatives for inputs include Turing computable numbers \cite{Turing,weihrauch2000} and real or complex numbers \cite{BSS,BCSS}. Alternatives for computation include the Pushdown Automaton \cite{Sipser}, the Blum--Shub--Smale Machine \cite{BSS,BCSS}, and the Quantum Turing Machine \cite{Deutsch}. Different models of computation can solve different tasks.

\item \textit{Model of Complexity}: What is the cost of a computation? In this article, we use \textbf{time complexity} measured in units of bit operations; i.e., the number of  \textsc{read}, \textsc{write}, \textsc{move}, and other tape-level instructions on bits. This is the same for the $\varepsilon$-accuracy complexity model\footnote{While the $\varepsilon$-accuracy complexity model is more realistic for numerical computations, it is not based on the IEEE floating-point standards \cite{kahan,overton}. On the other hand, a model that combines both the flexibility of the $\varepsilon$-accuracy complexity model and the reality of floating-point arithmetic would inevitably be enormously complicated \cite[Section~2.4]{Vavasis}.} \cite{HochbaumShanti}. In the Blum--Cucker--Shub--Smale (BCSS)  model, it is time complexity measured in units of arithmetic and branching operations involving inputs of real or complex numbers. In quantum computing, it is time complexity measured in units of unitary operations on qubits. There are yet other models of complexity that measure other types of computational costs. For example, complexity in the Valiant model is based on arithmetic circuit size.

\item \textit{Model of Reducibility}: Which problems do we consider equivalent in hardness? For us, it is the \textbf{Cook--Karp--Levin} (CKL) sense of reducibility \cite{Cook,Karp,Levin} and its corresponding problem classes: P, NP, NP-complete, NP-hard, etc. Reducibility in the BCSS model is essentially based on CKL. There is also reducibility in the Valiant sense, which applies to the aforementioned Valiant model and gives rise to the complexity classes $\mathit{VP}$ and $\mathit{VNP}$ \cite{Valiant,Burgisser}.
\end{enumerate}

Computability is a question to be answered in Level I, whereas difficulty is to be answered in Levels II and III. In Level II, we have restricted ourselves to time complexity since this is the most basic measure and it already reveals that tensor problems are hard.  In Level III, there is strictly speaking a subtle difference between the definition of reducibility by Cook \cite{Cook} and that by Karp and Levin \cite{Karp,Levin}. We define precisely our notion of reducibility in Section \ref{NPhardness}.


Before describing our model more fully, we recall the well-known Blum--Cucker--Shub--Smale framework for studying complexity of real and complex computations \cite{BSS,BCSS}.  In this model, an input is a list of $n$ real or complex numbers, without regard to how they are represented.  In this case, algorithmic computation (essentially) corresponds to arithmetic and branching on equality using a finite number of states, and a measure of computational complexity is the number of these basic operations\footnote{To illustrate the difference between BCSS/CKL, consider the problem of deciding whether two integers $r$, $s$ multiply to give an integer $t$.  For BCSS, the time complexity is constant since one can compute $u := rs-t$ and check ``$u = 0?$" in constant time.  Under CKL, however, the problem has best-known time complexity of $N\log(N) 2^{O(\log^*{N})}$, where $N$ is the number of bits to specify $r$, $s$, and $t$ \cite{Furer,De}.} needed to solve a problem as a function of $n$.  The central message of our paper is that many problems in linear algebra that are efficiently solvable on a Turing machine become NP-hard in multilinear algebra.  Under the BCSS model, however, this distinction is not yet possible. For example, while it is well-known that the feasibility of a linear program is in $\mathit{P}$ under the traditional CKL notion of complexity \cite{Khachiyan}, the same problem studied within  BCSS is among the most daunting open problems in Mathematics (it is the 9th ``Smale Problem" \cite{smale19}). The BCSS model has nonetheless produced significant contributions to computational mathematics, especially to the theory of polynomial equation solving (e.g., see \cite{beltpardo} and the references therein).  

We now explain our model of computation.  All computations are assumed to be performed on a Turing machine \cite{Turing} with the standard notion of time complexity involving operations on bits.  Inputs will be rational numbers\footnote{The only exception is when we prove NP-hardness of symmetric tensor eigenvalue (Section \ref{symm_tensor_eig}), where we allow input eigenvalues $\lambda$ to be in the field $\mathbb F = \{a+b \sqrt{d}: a,b \in \mathbb Q\}$ for any fixed positive integer $d$. Note that such inputs may also be specified with a finite number of bits.} and specified by finite strings of bits.  Outputs will consist of rational numbers or \textsc{yes}/\textsc{no} responses.   A decision problem is said to be \textit{computable} (or \textit{decidable}) if there is a Turing machine that will output the correct answer (\textsc{yes}/\textsc{no}) for all allowable inputs in finitely many steps. It is said to be  \textit{uncomputable} (or \textit{undecidable}) otherwise; i.e., no Turing machine could always determine the correct answer in finitely many steps. Note that the definition of a problem includes a specification of allowed inputs. We refer the reader to \cite{Sipser} for a proper treatment and to \cite{Poo} for an extensive list of undecidable problems arising from many areas of modern mathematics.

Although quantities such as eigenvalues, spectral norms, etc., of a tensor will in general not be rational, we note our reductions have been carefully constructed such that they are rational (or at least finite bit-length) in the cases we study.

The next subsection describes our notion of reducibility for tensor decision problems such as Problem \ref{TensorZeroEvalue}  encountered above.

\subsection{NP-hardness}\label{NPhardness}

The following is the notion of NP-hardness that we shall use throughout this paper.  As described above, inputs will be rational numbers, and input size is measured in the number of bits required to specify the input. Briefly, we say that a problem $\mathscr{D}_1$ is \textit{polynomially reducible} to a problem $\mathscr{D}_2$ if the following holds:  any input to $\mathscr{D}_1$ can be transformed in polynomially many steps (in the input size) into a set of polynomially larger inputs to $\mathscr{D}_2$ problems such that the corresponding answers can be used to correctly deduce (again, in a polynomial number of steps) the answer to the original $\mathscr{D}_1$ question. Informally, $\mathscr{D}_1$  polynomially reduces to $\mathscr{D}_2$  if there is a way to solve $\mathscr{D}_1$  by a deterministic polynomial-time algorithm when that algorithm is allowed to compute answers to instances of problem $\mathscr{D}_2$ in unit time. Note that the relation is not symmetric --- $\mathscr{D}_1$ polynomially reduces to $\mathscr{D}_2$ does not imply that $\mathscr{D}_2$ also reduces to $\mathscr{D}_1$.

By the Cook--Levin Theorem, if one can polynomially reduce any particular NP-complete problem to a problem  $\mathscr{D}$, then all NP-complete problems are so reducible to  $\mathscr{D}$. We call a decision problem \textit{NP-hard} if one can polynomially reduce any NP-complete decision problem (such as whether a graph is $3$-colorable) to it \cite{Kn1,Kn2}. Thus, an NP-hard problem is at least as hard as any NP-complete problem and quite possibly harder.\footnote{For those unfamiliar with these notions, we feel obliged to point out that the set of NP-hard problems is different than the set of NP-complete ones. First of all, an NP-hard problem may not be a decision problem, and secondly, even if we are given an NP-hard decision problem, it might not be in the class  $\mathit{NP}$; i.e., one might not be able to certify a \textsc{yes} decision in a polynomial number of steps.  
NP-complete problems, by contrast, are always decision problems and in the class $\mathit{NP}$.}

Although tensor eigenvalue is computable for $\mathbb{F}= \mathbb{R}$, it is nonetheless NP-hard, as our next theorem explains. For two sample reductions, see Example~\ref{3OL_ex} below. 

\begin{theorem}\label{TEvalueNP}
Graph $3$-colorability is polynomially reducible to tensor $0$-eigenvalue over $\mathbb{R}$.  Thus, deciding tensor eigenvalue over $\mathbb{R}$ is NP-hard.
\end{theorem}

A basic open question is whether deciding tensor eigenvalue is also NP-complete.  In other words, if a nontrivial solution to \eqref{l2eig} exists for a fixed $\lambda$, is there a polynomial-time verifiable certificate of this fact?  A natural candidate for the certificate is the eigenvector itself, whose coordinates would be represented as certain zeroes of univariate polynomials with rational coefficients.  The relationship between the size of these coefficients and the size of the input, however, is subtle and beyond our scope.  In the case of linear equations, polynomial bounds on the number of digits necessary to represent a solution can already be found in \cite{Edmonds} (see \cite[Theorem~2.1]{Vavasis}); and for the sharpest results to date on certificates for homogeneous linear systems, see \cite{freitas2011upper}. In contrast, rudimentary bounds for the type of tensor problems considered in this article are, as far as we know, completely out of reach.

\begin{example}[Real tensor $0$-eigenvalue solves $3$-colorability]\label{3OL_ex}
Let $G = (V,E)$ be a simple, undirected graph with vertices $V= \{1,\dots, v\}
$ and edges $E$.  Recall that a \textit{proper} (\textit{vertex}) $3$-\textit{coloring} of $G$ is an assignment of one of three colors to each of its vertices such that adjacent vertices receive different colors.  We say that $G$ is \emph{$3$-colorable} if it has a proper $3$-coloring.  See Fig.~\ref{smlgraph} for an example of a $3$-colorable graph on four vertices.  Determining whether $G$ is $3$-colorable is a well-known NP-complete decision problem.

As we shall see in Section~\ref{quadfeas}, proper $3$-colorings of the left-hand side graph in Fig.~\ref{smlgraph} can be encoded as the nonzero real solutions to the following square set of $n = 35$ quadratic polynomials in $35$ real unknowns $a_i,b_i,c_i,d_i$ ($i = 1,\dots,4$), $u$, $w_i$ ($i = 1,\dots,18$):

{\footnotesize
\begin{equation}\label{grapheqs}
\begin{split}
& a_{1}c_{1}-b_{1}d_{1}-u^2,  \ b_{1}c_{1}+a_{1}d_1, \ c_{1}u-a_1^2+b_1^2,  \ d_{1}u-2a_1 b_1, \ a_{1}u-c_1^2+d_1^2, \ b_{1}u-2d_1 c_1, \\
& a_{2}c_{2}-b_{2}d_{2}-u^2,  \ b_{2}c_{2}+a_{2}d_2, \ c_{2}u-a_2^2+b_2^2,  \ d_{2}u-2a_2 b_2, \  a_{2}u-c_2^2+d_2^2, \ b_{2}u-2d_2 c_2, \\
& a_{3}c_{3}-b_{3}d_{3}-u^2,  \ b_{3}c_{3}+a_{3}d_3 , \ c_{3}u-a_3^2+b_3^2,  \ d_{3}u-2a_3 b_3, \ a_{3}u-c_3^2+d_3^2, \ b_{3}u-2d_3 c_3, \\
& a_{4}c_{4}-b_{4}d_{4}-u^2,  \ b_{4}c_{4}+a_{4}d_4, \ c_{4}u-a_4^2+b_4^2,  \ d_{4}u-2a_4 b_4, \ a_{4}u-c_4^2+d_4^2,  \ b_{4}u-2d_4 c_4, \\
& a_{1}^2-b_{1}^2 + a_1a_3 -b_1b_3 + a_3^2 -b_3^2,  \ a_{1}^2-b_{1}^2 + a_1a_4 -b_1b_4 + a_4^2 -b_4^2, \ a_{1}^2-b_{1}^2 + a_1a_2 -b_1b_2 + a_2^2 -b_2^2,\\
& a_{2}^2-b_{2}^2 + a_2a_3 -b_2b_3 + a_3^2 -b_3^2, \ a_{3}^2-b_{3}^2 + a_3a_4 -b_3b_4 + a_4^2 -b_4^2,  \ 2a_{1}b_1+a_{1}b_2+a_2b_1+2a_2b_2, \\
& 2a_{2}b_2+a_{2}b_3+a_3b_2+2a_3b_3, \ 2a_{1}b_1+a_{1}b_3+a_2b_1+2a_3b_3, \ 2a_{1}b_1+a_{1}b_4+a_4b_1+2a_4b_4,\\
& 2a_{3}b_3+a_{3}b_4+a_4b_3+2a_4b_4, \ w_1^2 + w_2^2 +  \dots+w_{17}^2 +  w_{18}^2.
\end{split}
\end{equation}}

\begin{figure}[!t]
\begin{center}
\includegraphics[width=5in]{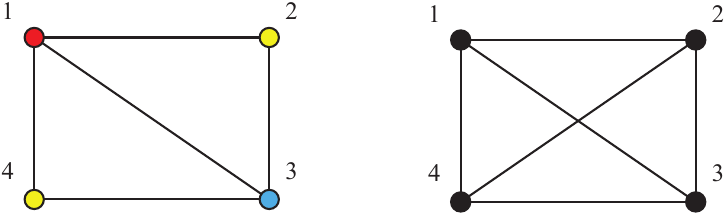}
\caption{Simple graphs with six proper $3$-colorings (graph at the left) or none (graph at the right).}
\label{smlgraph}
\end{center}
\end{figure}

Using symbolic algebra or numerical algebraic geometry software\footnote{For code used in this paper, see: \url{http://www.msri.org/people/members/chillar/code.html}.} (see the Appendix for a list), one can solve these equations to find six real solutions (without loss of generality, we may take $u=1$ and all $w_j = 0$), which correspond to the proper $3$-colorings of the graph $G$ as follows.  Fix one such solution and define $x_k := a_k + \mathrm{i} b_k \in \mathbb{C}$ for $k=1,\dots,4$ (we set $\mathrm{i} := \sqrt{-1}$).   By construction, these $x_k$ are one of the three cube roots of unity $\{1,\alpha,\alpha^2\}$ where $\alpha = \exp(2 \pi \mathrm{i}/3) = -\frac{1}{2} + \mathrm{i}  \frac{\sqrt{3}}{2}$ (see also  Fig.~\ref{approx_cube_roots}).

To determine a $3$-coloring from this solution, one ``colors" each vertex $i$ by the root of unity that equals $x_i$.  It can be checked that no two adjacent vertices share the same color in a coloring; thus, they are proper $3$-colorings.  For example, one solution is: 
\[ x_1 = -\frac{1}{2} - \mathrm{i} \frac{\sqrt{3}}{2}, \quad x_2 = 1, \quad x_3 =  -\frac{1}{2} + \mathrm{i} \frac{\sqrt{3}}{2}, \quad x_4 = 1.\]

Polynomials for the right-hand side graph in Fig.~\ref{smlgraph} are the same as \eqref{grapheqs} except for two additional ones encoding a new restriction for colorings, the extra edge $\{2,4\}$:
\[ a_{2}^2-b_{2}^2 + a_2a_4 -b_2b_4 + a_4^2 -b_4^2,  \quad 2a_{2}b_2+a_{2}b_4+a_4b_2+2a_4b_4.\]
One can verify with the same software that these extra equations force the system to have no nonzero real solutions, and thus no proper $3$-colorings.

Finally, note that since equivalence classes of eigenvectors correspond to proper $3$-colorings in this reduction, if we could count real (projective) eigenvectors with eigenvalue $\lambda = 0$, we would solve the enumeration problem for proper $3$-colorings of graphs (in particular, this proves Corollary~\ref{sharp_p_3col} below).
\qed
\end{example}

\subsection{Approximation Schemes}

Although the tensor eigenvalue decision problem is NP-hard and the eigenvector enumeration problem \#P-hard, it might be possible to approximate \textit{some} eigenpair ($\lambda,\mathbf{x}$) efficiently, which is important for some applications (e.g., \cite{SchultzS}). For those unfamiliar with these ideas, an \textit{approximation scheme} for a tensor problem (such as finding a tensor eigenvector) is an algorithm producing a (rational) approximate solution to within $\varepsilon$ of some solution. An approximation scheme is said to run in \textit{polynomial time} (PTAS) if its running-time is polynomial in the input size for any fixed $\varepsilon > 0$, and \textit{fully polynomial time} (FPTAS) if its running-time is polynomial in both the input size and $1/\varepsilon$.  There are other notions of approximation \cite{Hochbaum,Vazirani}, but we  limit our discussion to these. 

Fix $\lambda$, which we assume is an eigenvalue of a tensor $\mathcal{A}$.  Formally, we say it is \textit{NP-hard to approximate an eigenvector} of an eigenpair $(\lambda,\mathbf{x})$ to within $\varepsilon > 0$ if (unless $\mathit{P} = \mathit{NP}$) there is no polynomial-time algorithm that always produces a nonzero $\hat{\mathbf{x}} = [\hat{x}_1,\dots,\hat{x}_n]^{\top} \in \mathbb{F}^n$ that approximates some solution $\mathbf{0} \neq \mathbf{x} \in \mathbb{F}^n$ to system \eqref{l2eig} by satisfying for all $i$ and $j$:
\begin{equation}\label{approx_def}
 \lvert \hat{x}_i/\hat{x}_j - x_i/x_j \rvert < \varepsilon, \quad \text{whenever $x_j \neq 0$.}
\end{equation}
This measure of approximation is natural in our context because of the scale invariance of eigenpairs, and it is closely related to standard relative error $\lVert\hat{\mathbf{x}}-\mathbf{x}\rVert_{\infty}/\lVert \mathbf{x}\rVert_{\infty}$ in numerical analysis. We shall prove the following inapproximability result in Section~\ref{sec:EVP}.

\begin{theorem}\label{approx_evect}
It is NP-hard to approximate tensor eigenvector over $\mathbb{R}$ to  $\varepsilon = \frac{3}{4}$.
\end{theorem}

\begin{corollary}
No PTAS  approximating tensor eigenvector exists unless $\mathit{P} = \mathit{NP}$.
\end{corollary}

\subsection{Tensor Singular Value, Hyperdeterminant, and Spectral Norm}

Our next result involves the singular value problem. We postpone definitions until Section~\ref{sec:SVP}, but state the main result here. 


\begin{theorem}\label{sing_nphard_thm}
Let $\mathbb{F} = \mathbb{R}$ or $\mathbb{C}$, and fix $\sigma \in \mathbb{Q}$.  Deciding whether $\sigma$ is a singular value over $\mathbb{F}$ of a tensor is NP-hard.
\end{theorem}

There is also a notion of  \textit{hyperdeterminant}, which we discuss in more depth in Section~\ref{bilineareqs}  (see also Section~\ref{sec:combdet} for a different notion of hyperdeterminant). Like the determinant, this is a homogeneous polynomial (with integer coefficients) in the entries of a tensor that vanishes if and only if the tensor has a zero singular value.  The following problem is important for multilinear equation solving (e.g., Example~\ref{neq2Hyperdet}).

\begin{problem}\label{Hypdet}
Decide if the hyperdeterminant of a tensor is zero.
\end{problem}

We were unable to determine the hardness of Problem~\ref{Hypdet}, but conjecture that it is difficult.

\begin{conjecture}\label{hyp_conj}
It is NP-hard to decide the vanishing of the hyperdeterminant.
\end{conjecture}

We are, however, able to evaluate the complexity of computing the spectral norm (see Definition \ref{defn:spec_norm}), which is a special singular value of a tensor.
 
\begin{theorem}\label{spec_norm_nphard_thm}
Fix any nonzero $\sigma \in \mathbb{Q}$.  Deciding whether $\sigma$ is the spectral norm of a tensor is NP-hard.
\end{theorem}

Determining the spectral norm is an optimization (maximization) problem. Thus, while it is NP-hard to decide tensor spectral norm, there might be efficient ways to approximate it.  A famous example of approximating solutions to problems whose decision formulations are NP-hard is the classical result of \cite{goemans1995} which gives a polynomial-time algorithm to determine a cut size of a graph that is at least $.878$ times that of a maximum cut.   In fact, it has been shown, assuming the Unique Games Conjecture, that Goemans-Williamson's approximation factor is best possible \cite{KhotGW2007}.  We refer the reader to  \cite{alon2004,bachoc2008,bachoc2009,briet2010positive,briet2010grothendieck,HLZ} for some recent work in the field of approximation algorithms.

Formally, we say that it is \textit{NP-hard to approximate the spectral norm} of a tensor to within $\varepsilon > 0$ if (unless $\mathit{P}=\mathit{NP}$) there is no polynomial-time algorithm giving a guaranteed lower bound for the spectral norm that is at least a $(1-\varepsilon)$-factor of its true value. Note that $\varepsilon$ here might be a function of the input size.  A proof of the following can be found in Section~\ref{sec:SVP}.

\begin{theorem}\label{approx_spec_nphard_thm}
It is NP-hard to approximate the spectral norm of a tensor $\mathcal{A}$ to within \[\varepsilon = 1- \left(1+\frac{1}{N(N-1)}\right)^{-1/2} = \frac{1}{2N(N-1)} + O(N^{-4}),\] where $N$ is the input size of $\mathcal{A}$.
\end{theorem}

\begin{corollary}\label{FPTAS_spec_nphard_thm}
No FPTAS to approximate spectral norm exists unless $\mathit{P} = \mathit{NP}$.
\end{corollary}
\begin{proof}
Suppose there is a FPTAS for the tensor spectral norm problem and take $\varepsilon = 1/(4N^2)$ as the approximation error desired for a tensor of input size $N$.  Then, in time polynomial in $1/\varepsilon = 4N^2$ (and thus in $N$), it would be possible to approximate the spectral norm of a tensor with input size $N$ to within $1- \left(1+\frac{1}{N(N-1)}\right)^{-1/2}$ for all large $N$.  From Theorem~\ref{approx_spec_nphard_thm}, this is only possible if $\mathit{P} = \mathit{NP}$.
\end{proof}

\subsection{Tensor Rank}

The \textit{outer product} $\mathcal{A} = \mathbf{x} \otimes \mathbf{y} \otimes \mathbf{z}$ of  vectors $\mathbf{x}\in \mathbb{F}^l$, $\mathbf{y}\in \mathbb{F}^m$, and $\mathbf{z}\in \mathbb{F}^n$ is the tensor $\mathcal{A} = \llbracket a_{ijk}\rrbracket_{i,j,k=1}^{l,m,n}$ given by $a_{ijk} = x_iy_jz_k$.
A nonzero tensor that can be expressed as an outer product of vectors is called \textit{rank}-$1$.  
More generally, the \textit{rank} of a tensor \mbox{$\mathcal{A}=\llbracket a_{ijk}\rrbracket
\in\mathbb{F}^{l\times m\times n}$}, denoted $\operatorname{rank}(\mathcal{A})$, is the minimum $r$ for which $\mathcal{A}$ is a sum of $r$ rank-$1$ tensors \cite{Hi1,Hi2} with $\lambda_i \in \mathbb{F}$, $\mathbf{x}_{i}\in \mathbb{F}^l$,
$\mathbf{y}_{i}\in \mathbb{F}^m$, and $\mathbf{z}_{i}\in \mathbb{F}^n$:
\begin{equation}
\operatorname{rank}(\mathcal{A}):=\min\Bigl\{r : \mathcal{A}%
=\sum\nolimits_{i=1}^{r}\lambda_{i}\,\mathbf{x}_{i}%
\otimes\mathbf{y}_{i}\otimes\mathbf{z}_{i}\Bigr\}. \label{rank}%
\end{equation}

For a symmetric tensor $\mathcal{S} \in \mathbb{F}^{n \times n \times n}$, we shall require that the vectors in the outer product be the same:
\begin{equation}
\operatorname{srank}(\mathcal{S}):=\min\Bigl\{r : \mathcal{S}%
=\sum\nolimits_{i=1}^{r}\lambda_{i}\,\mathbf{x}_{i}
\otimes\mathbf{x}_{i}\otimes\mathbf{x}_{i}\Bigr\}. \label{srank}
\end{equation}
The number in \eqref{srank} is called the \textit{symmetric rank} of $\mathcal{S}$. It is still not known whether a symmetric tensor's symmetric rank is always its rank (this is the Comon Conjecture \cite{Landsberg}; see also \cite{reznick2010}), although the best symmetric rank-$1$ approximation and the best rank-$1$ approximation coincide (see Section~\ref{sec:symmproblems}).  Note that these definitions of rank agree with matrix rank when applied to a $2$-tensor. 

Our next result says that approximating a tensor with a single rank-$1$ element is already hard. A consequence is that data analytic models under the headings of \textsc{parafac}, \textsc{candecomp}, and \textsc{tucker} --- originating from psychometrics but having newfound popularity in other areas of data analysis --- are all NP-hard to fit  in the simplest case.

\begin{theorem}\label{rank1approx_nphard_thm}
Rank-$1$ tensor approximation is NP-hard.
\end{theorem}


As will become clear, tensor rank as defined in \eqref{rank} implicitly depends on the choice of field. Suppose that $\mathbb{F} \subseteq \mathbb{E}$ is a subfield of a field $\mathbb{E}$.  If $\mathcal{A}\in \mathbb{F}^{l\times m\times n}$ is as in (\ref{rank}), but we allow $\lambda_i \in \mathbb{E}$, $\mathbf{x}_{i}\in \mathbb{E}^l$, $\mathbf{y}_{i}\in \mathbb{E}^m$, and $\mathbf{z}_{i}\in \mathbb{E}^n$, then the number computed in (\ref{rank}) is called the \textit{rank of $\mathcal{A}$ over $\mathbb E$}. We will write $\operatorname{rank}_{\mathbb{E}}(\mathcal{A})$ (a notation that we will use whenever the choice of field is important) for the rank of $\mathcal{A}$ over $\mathbb{E}$. In general, it is possible that
\[
\operatorname{rank}_{\mathbb{E}}(\mathcal{A}) < \operatorname{rank}_{\mathbb{F}}(\mathcal{A}).
\]
We discuss this in detail in Section~\ref{sec:Rank}, where we give a new result about the rank of tensors over changing fields (in contrast, the rank of a matrix does not change when the ground field is enlarged).  The proof uses symbolic and computational algebra in a fundamental way.

\begin{theorem}\label{rankQRthm}
There is a rational tensor $\mathcal{A} \in \mathbb{Q}^{2 \times 2 \times 2}$ with $\operatorname{rank}_{\mathbb{R}}(\mathcal{A}) < \operatorname{rank}_{\mathbb{Q}}(\mathcal{A})$.
\end{theorem}

H\aa stad has famously shown that tensor rank over $\mathbb{Q}$ is NP-hard \cite{Haa}. Since tensor rank over $\mathbb{Q}$ differs in general from tensor rank over $\mathbb{R}$, it is natural to ask if tensor rank might still be NP-hard over $\mathbb{R}$ and $\mathbb{C}$. In Section~\ref{sec:Rank}, we shall explain how the argument in \cite{Haa} also proves the following.

\begin{theorem}[H\aa stad]
Tensor rank is NP-hard over $\mathbb{R}$ and $\mathbb{C}$.
\end{theorem}

\subsection{Symmetric Tensors}

One may wonder if NP-hard problems for general nonsymmetric tensors might perhaps become tractable for symmetric ones. We show that restricting these problems to the class of symmetric tensors does not remove NP-hardness. As with their nonsymmetric counterparts, eigenvalue, singular value, spectral norm, and best rank-$1$ approximation problems for symmetric tensors all remain NP-hard. In particular, the NP-hardness of symmetric spectral norm in Theorem~\ref{thm:symmresults} answers an open problem in \cite{BV}.

\subsection{\#P-hardness and VNP-hardness}

As is evident from the title of our article and the list in Table~\ref{tab:tensor_complexity}, we have used NP-hardness as our primary measure of computational intractability. Valiant's notions of \textit{\#P-completeness} \cite{Valiant79} and \textit{VNP-completeness} \cite{Valiant} are nonetheless  relevant to tensor problems. We prove the following result about tensor eigenvalue over $\mathbb{R}$ (see Example~\ref{3OL_ex}).
\begin{corollary}\label{sharp_p_3col}
It is \#P-hard to count tensor eigenvectors over $\mathbb{R}$. 
\end{corollary}
Because of space constraints, we will not elaborate on  these notions except to say that \#P-completeness applies to enumeration problems associated with NP-complete decision problems while VNP-completeness applies to polynomial evaluation problems.  For example, deciding whether a graph is $3$-colorable is an NP-complete decision problem, but counting the number of proper $3$-colorings is a \#P-complete enumeration problem. 

The VNP complexity classes involve questions about the minimum number of arithmetic operations required to evaluate multivariate polynomials. An illuminating example \cite[Example~13.3.1.2]{Landsberg} is given by
\[
p_n(x,y) = x^n + nx^{n-1}y + \binom{n}{2}x^{n-2}y^2 +\dots +  \binom{n}{2}x^{2}y^{n-2}  + nxy^{n-1}+ y^n,
\]
which at first glance requires $n(n+1)$ multiplications and $n$ additions to evaluate.  However, from the binomial expansion, we have $p_n(x,y) =(x+y)^n$, and so the operation count can be reduced to $n-1$ multiplications and $1$ addition. In fact, this last count is not minimal, and a more careful study of such questions involves arithmetic circuit complexity.  We refer the reader to \cite{Burgisser} for a detailed exposition.

As in Section~\ref{NPhardness}, one may also analogously define notions of  \#P-hardness and VNP-hardness: A problem is said to be \textit{\#P-hard} (resp.\ \textit{VNP-hard}) if every \#P-complete enumeration problem (resp.\ VNP-complete polynomial evaluation problem)  may be polynomially reduced to it. Evidently, a \#P-hard (resp.\ VNP-hard) problem is at least as hard and quite possibly harder than a \#P-complete (resp.\ VNP-complete) problem.\fin

\subsection{Intractable Matrix Problems}

Not all matrix problems are tractable. For example, matrix $(p,q)$-norms when $1\le q \le p \le \infty$ \cite{Stein}, nonnegative rank of nonnegative matrices \cite{Vavasis1,kannan2012}, sparsest null vector \cite{ColePot}, and rank minimization with linear constraints \cite{Nata} are all known to be NP-hard; the matrix $p$-norm when $p \ne 1, 2, \infty$ is NP-hard to approximate \cite{HO}; and evaluating the permanent of a $\{0,1\}$-valued matrix is a well-known \#P-complete problem \cite{Valiant79}. Our intention is to highlight the sharp distinction between the computational intractability of  certain tensor problems and the tractability of their matrix specializations. As such we do not investigate tensor problems that are known to be hard for matrices. 

\subsection{Quantum Computers}

Another question sometimes posed to the authors is whether quantum computers might help with these problems. This is believed unlikely because of the seminal works \cite{BV2} and \cite{FR} (see also the survey \cite{Fort}). These authors have demonstrated that the complexity class of bounded error quantum polynomial time (BQP) is not expected to overlap with the complexity class of NP-hard problems. Since BQP encompasses the decision problems solvable by a quantum computer in polynomial time, the NP-hardness results in this article show that quantum computers are unlikely to be effective for tensor problems.

\subsection{Finite Fields}

We have restricted our discussion in this article to extension fields of $\mathbb{Q}$ as these are most relevant for the numerical computation arising in science and engineering. Corresponding results over finite fields are nonetheless also of interest in computer science; for instance, quadratic feasibility arises in cryptography \cite{CGMT} and tensor rank arises in boolean satisfiability problems \cite{Haa}.


\section{Quadratic feasibility is NP-hard}\label{quadfeas}

Since it will be a basic tool for us in proving results about tensors (e.g., Theorem~\ref{TEvalueNP} for tensor eigenvalue), we examine the complexity of solving quadratic equations.

\begin{problem}\label{QuadFeasDef} 
Let $\mathbb{F}= \mathbb{Q}$, $\mathbb{R}$, or $\mathbb C$.  For $i = 1,\dots,m$, let $A_{i}\in \mathbb{Q}^{n \times n}$, 
$\mathbf{b}_i \in \mathbb{Q}^n$, and $c_i \in \mathbb{Q}$.  Also, let $\mathbf{x}=[x_{1},\dots,x_{n}]^{\top}$ be a vector of unknowns, and set $G_{i}(\mathbf{x})=\mathbf{x}^{\top}
A_{i}\mathbf{x} + \mathbf{b}_i^{\top}\mathbf{x} + c_i$.  Decide if the system $\{G_{i}(\mathbf{x})=0\}_{i=1}^{m} $ has
a solution $\mathbf{x}\in\mathbb{F}^{n}$.
\end{problem}

Another quadratic problem that is more natural in our setting is the following.

\begin{problem}[Quadratic Feasibility]\label{HomQuadFeasDef} 
Let $\mathbb{F}= \mathbb{Q}$, $\mathbb{R}$, or $\mathbb{C}$.  For $i = 1,\dots,m$, let $A_{i} \in \mathbb{Q}^{n \times n}$
and set $G_{i}(\mathbf{x})=\mathbf{x}^{\top}
A_{i}\mathbf{x}$.  Decide if the system of equations $\{G_{i}(\mathbf{x})=0\}_{i=1}^{m} $ has
a   nonzero solution $\mathbf{0} \neq \mathbf{x}\in\mathbb{F}^{n}$.
\end{problem}
\begin{remark}\label{real_hom_rmk}
It is elementary that the (polynomial) complexity of Problem~\ref{QuadFeasDef} is the same as that of Problem~\ref{HomQuadFeasDef} when $\mathbb{F} = \mathbb{Q}$ or $\mathbb{R}$.  To see this, homogenize each equation $G_{i} = 0$ in Problem \ref{QuadFeasDef} by introducing $z$ as a new unknown: $\mathbf{x}^{\top}A_{i}\mathbf{x} + \mathbf{b}_i^{\top}\mathbf{x}z + c_iz^2 = 0$.  Next, introduce the quadratic equation $x_1^2+\cdots + x_n^2- z^2 = 0$.  This new set of equations is easily seen to have a nonzero solution if and only if the original system has any solution at all.  The main trick used here is that $\mathbb R$ is a \textit{formally real field}: that is, we always have $\sum x_i^2 = 0 \ \Rightarrow \ x_i = 0$ for all $i$.
\end{remark}

Problem~\ref{HomQuadFeasDef} for $\mathbb{F}= \mathbb{R}$ was studied in \cite{Barv}.  There, it is shown that for fixed $n$, one can decide the real feasibility of $m \gg n$ such quadratic equations  $\{G_{i}(\mathbf{x})=0\}_{i=1}^{m}$ in $n$ unknowns in a number of arithmetic operations that is polynomial in $m$.  In contrast, we shall show that quadratic feasibility is NP-hard over $\mathbb{R}$ and $\mathbb{C}$. 

To give the reader a sense of the generality of nonlinear equation solving, we first explain, in very simplified form, the connection of quadratic systems to the Halting Problem established in the seminal works \cite{Davis61,matij70}. Collectively, these papers resolve (in the negative) Hilbert's $10$th Problem \cite{hilbert1902}: Is there a finite procedure to decide the solvability of general polynomial equations over the integers (the so-called \textit{Diophantine Problem} over $\mathbb Z$).  For an exposition of the ideas involved, see \cite{davis1976} or the introductory book \cite{matijbook}.


The following fact in theoretical computer science is a basic consequence of these papers.  Fix a universal Turing machine.  There is a listing of all Turing machines $\mathscr{T}_x$ ($x = 1,2,\dots .$) and a finite set $S = S(x)$ of integral quadratic equations in the parameter $x$ and other unknowns with the following property: For each particular positive integer $x = 1,2,\dots ,$ the system $S(x)$ has a solution in positive integers if and only if Turing machine $\mathscr{T}_x$ halts with no input.  In particular, polynomial equation solving over the positive integers is undecidable.  

\begin{theorem}[Davis--Putnam--Robinson, Matijasevic]\label{QFQ_undec_thm}
Problem~\ref{QuadFeasDef} is undecidable over $\mathbb{Z}$.
\end{theorem}
\begin{proof}
Consider the above system $S$ of polynomials.  For each of the unknowns $y$ in $S$, we add an equation in four new variables $a,b,c,d$ encoding (by Lagrange's Four-Square Theorem) that $y$ should be a positive integer:  
\[y = a^2 + b^2 + c^2 + d^2 + 1.\] 
Let $S'$ denote this new system.  Polynomials $S'(x)$ have a common zero in integers if and only if $S(x)$ has a solution in positive integers.  Thus, if we could decide the solvability of quadratic equations over the integers, we would solve the Halting problem.
\end{proof}

\begin{remark}
Using \cite[pp.~552]{jones80}, one can construct an explicit set of quadratics whose solvability over $\mathbb Z$ encodes whether a given Turing machine halts.
\end{remark}

The decidability of Problem \ref{QuadFeasDef} with $\mathbb F = \mathbb Q$ is still unknown, as is the general Diophantine problem over $\mathbb Q$.  See \cite{poonen2003} for some  progress on this hard problem.

While a system of quadratic equations of the form in  Problem~\ref{QuadFeasDef} determined by coefficients $\mathcal{A} \in \mathbb{Z}^{m \times n \times n}$ is undecidable, we may decide whether a system of \textit{linear} equations $A\mathbf{x} = \mathbf{b}$ with $A\in \mathbb{Z}^{m \times n}$, $\mathbf{b} \in \mathbb Z^m$ has an integral solution $\mathbf{x}$ by computing the Smith normal form of the coefficient matrix $A$ (e.g., see \cite{yap2000fundamental}). We view this as another instance where the transition from matrices $A$ to tensors $\mathcal{A}$ has a drastic effect on computability.  Another example of a matrix problem that becomes undecidable when one states its analogue for $3$-tensors is given in Section~\ref{sec:biv}.

We next study quadratic feasibility when $\mathbb{F} = \mathbb{R}$ and $\mathbb{C}$ and show that it is NP-hard.  Variations of this basic result appear in \cite{bayer,lovasz,GrenetKoiranPortier}.

\begin{theorem}\label{QFHard}
Let $\mathbb{F} = \mathbb{R}$ or $\mathbb{C}$.  Graph $3$-colorability is polynomially reducible to quadratic feasibility over $\mathbb{F}$.  Thus, Problem~\ref{HomQuadFeasDef} over $\mathbb{F}$ is NP-hard.
\end{theorem}

The idea of turning colorability problems into questions about polynomials appears to originate
with Bayer's thesis although it has arisen in several other
places, including \cite{lovasz,deloera,deloera08}. For a recent application of polynomial algebra to deciding unique $3$-colorability, see \cite{HW}.

To prove Theorem~\ref{QFHard}, we shall reduce graph $3$-colorability to quadratic feasibility over $\mathbb{C}$. The result for $\mathbb{F} = \mathbb{R}$ then follows from the following fact.

\begin{lemma}
\label{complexrealquad} Let $\mathbb{F} = \mathbb{R}$ and $A_{i}$, $G_i$ be as in
Problem~\ref{HomQuadFeasDef}. Consider a new system $H_{j}(\mathbf{x}
)=\mathbf{x}^{\top}B_{j}\mathbf{x}$ of $2m$ equations in $2n$ unknowns
given by:
\[
B_{i}=
\begin{bmatrix}
A_{i} & 0\\
0 & -A_{i}
\end{bmatrix}
,\quad B_{m+i}=
\begin{bmatrix}
0 & A_{i}\\
A_{i} & 0
\end{bmatrix}
,\quad i=1,\dots,m.
\]
The equations $\{H_{j}(\mathbf{x})=0\}_{j=1}^{2m}$ have a nonzero real
solution $\mathbf{x}\in\mathbb{R}^{2n}$ if and only if the
equations $\{G_{i}(\mathbf{z})=0\}_{i=1}^{m}$ have a nonzero complex
solution $\mathbf{z}\in\mathbb{C}^{n}$.
\end{lemma}

\begin{proof}
By construction, a nonzero solution $\mathbf{z}=\mathbf{u}+ \mathrm{i} \mathbf{v}
 \in\mathbb{C}^{n}$ with $\mathbf{u},\mathbf{v} \in\mathbb{R}^{n}$ to equations 
\mbox{$\{G_{i}(\mathbf{z})=0\}_{i=1}^{m}$} corresponds to a nonzero real solution
$\mathbf{x}=[\mathbf{u}^{\top},\mathbf{v}^{\top}]^{\top}$ to $\{H_{j}(\mathbf{x})=0\}_{j=1}^{2m}$.
\end{proof}




%

The trivial observation below gives flexibility in specifying quadratic feasibility problems over $\mathbb{R}$.
This is useful since the system defining an eigenpair is a square system.

\begin{lemma}\label{quadfeasmoreeqvars} 
Let $G_{i}(\mathbf{x})=\mathbf{x}^{\top}%
A_{i}\mathbf{x}$ for $i=1,\dots,m$ with $A_{i}\in \mathbb{R}^{n \times n}$. Consider a new system $H_{i}(\mathbf{x})= \mathbf{x}^{\top}B_{i}\mathbf{x}$ of $r \geq m+1$ equations in $s \geq n$ unknowns given by $s \times s$ matrices:
\[
B_{i}=
\begin{bmatrix}
A_{i} & 0\\
0 & 0
\end{bmatrix}
, \ i=1,\dots,m; \quad B_{j}=
\begin{bmatrix}
0 & 0\\
0 & 0
\end{bmatrix}
, \ j=m+1,\dots,r-1; \quad B_{r} =
\begin{bmatrix}
0 & 0\\
0 & I
\end{bmatrix}
;
\]
in which $I$ is the $(s-n) \times(s-n)$ identity matrix. Equations
$\{H_{i}(\mathbf{x})=0\}_{i=1}^{r}$ have a nonzero solution
$\mathbf{x}\in\mathbb{R}^{s}$ if and only if $\{G_{i}(\mathbf{x})=0\}_{i=1}^{m}$ have a nonzero solution
$\mathbf{x}\in\mathbb{R}^{n}$.
\end{lemma}

The following set of polynomials $C_{G}$ allows us to relate feasibility of a polynomial
system to $3$-colorability of a graph $G$.  An instance of this encoding (after applying Lemmas~\ref{complexrealquad} and \ref{quadfeasmoreeqvars} appropriately) is Example~\ref{3OL_ex}.

\begin{definition}
The \textit{color encoding} of a graph $G=(V,E)$ with $v$ vertices is the set of $4v$
quadratic polynomials in $2v+1$ unknowns $x_1,\ldots,x_v,y_1,\ldots,y_v,z$:
\begin{equation}
\label{CGdef}C_{G}:=%
\begin{cases}
x_{i}y_{i}-z^{2},\quad y_{i}z-x_{i}^{2},\quad x_{i}z-y_{i}^{2}, &
i=1,\dots,v,\\
\sum_{j:\{i,j\} \in E} (x_{i}^{2} + x_{i} x_{j} + x_{j}^{2}), & i=1,\dots,v.
\end{cases}
\end{equation}
\end{definition}

\begin{lemma}
\label{colorencodinglemma} $C_{G}$ has a nonzero
complex solution if and only if  $G$ is $3$-colorable.
\end{lemma}

\begin{proof}
Suppose that $G$ is $3$-colorable and let $[x_{1},\dots,x_{v}]^\top \in
\mathbb{C}^{n}$ be a proper $3$-coloring of $G$, encoded using cube roots of unity as in Example~\ref{3OL_ex}.
Set $z = 1$ and $y_{i} = 1/x_{i}$ for $i = 1,\dots,v$; we claim that these
numbers are a common zero of $C_{G}$. It is clear that the first $3v$
polynomials in \eqref{CGdef} evaluate to zero. Next consider any expression of
the form $p_i =  \sum_{j:\{i,j\} \in E} x_{i}^{2} + x_{i} x_{j} + x_{j}^{2}$. Since
we have a $3$-coloring, $x_{i} \neq x_{j}$ for $\{i,j\} \in E$; thus,
\[
0 = x_{i}^{3}-x_{j}^{3} = \frac{x_{i}^{3}-x_{j}^{3}}{x_{i} - x_{j}} =
x_{i}^{2} + x_{i} x_{j} + x_{j}^{2}.
\]
In particular, each $p_i$ evaluates to zero as desired.  

Conversely, suppose that the polynomials $C_{G}$ have a common
nontrivial solution,
\[
\mathbf{0} \neq [x_{1},\dots,x_{v},y_{1},\dots,y_{v},z]^\top \in\mathbb{C}^{2v+1}.
\]
If $z=0$, then all of the $x_{i}$ and $y_{i}$ must be zero as well. Thus  $z\neq0$, and since the equations are homogenous, we may assume that our
solution has $z=1$. It follows that $x_{i}^{3}=1$ for all $i$ so that
$[x_{1},\dots,x_{v}]^{\top}$ is a $3$-coloring of $G$. We are left with verifying
that it is proper. If $\{i,j\}\in E$ and $x_{i}=x_{j}$, then $x_{i}^{2}+x_{i}x_{j}+x_{j}^{2}=3x_{i}^{2}$; otherwise, 
if $\{i,j\} \in E$ and $x_{i} \neq x_{j}$, then $x_{i}^{2}+x_{i}x_{j}+x_{j}^{2} = 0$.  Thus, $p_i = 3r_ix_i^2$, where $r_i$ is the number of vertices $j$ 
adjacent to $i$ that have $x_{i}=x_{j}$.  It follows that $r =0$ so that $x_{i}\neq x_{j}$ for all $\{i,j\}\in E$, and thus $G$ has a proper $3$-coloring.
\end{proof}

We close with a proof that quadratic feasibility over $\mathbb{C}$ (and therefore, $\mathbb{R}$) is NP-hard.

\begin{proof}
[of Theorem~\ref{QFHard}]Given a graph $G$, construct the color encoding polynomials $C_{G}$. From Lemma~\ref{colorencodinglemma}, the homogeneous quadratics $C_{G}$ have a nonzero complex solution if and only if $G$ is $3$-colorable. Thus, solving Problem~\ref{HomQuadFeasDef} over $\mathbb{C}$ in polynomial time would allow us to do the same for graph $3$-colorability.
\end{proof}




%

\section{Bilinear system is NP-hard}\label{bilineareqs}

We next consider some natural bilinear extensions to the quadratic
feasibility problems encountered earlier, generalizing Example~\ref{neq2Hyperdet} below. The main result is
Theorem~\ref{zerosingvaluethm}, which shows that the following 
feasibility problem over $\mathbb{R}$ or $\mathbb{C}$ is NP-hard. In Section~\ref{sec:SVP}, we use this  to show that certain singular value problems are also NP-hard.

\begin{problem}
[Tensor Bilinear Feasibility]\label{tripletensorfeas} Let $\mathbb{F} = \mathbb{R}$ or $\mathbb{C}$. Let $\mathcal{A} =
\llbracket  a_{ijk} \rrbracket \in\mathbb{Q}^{l \times m \times n}$, and set $A_{i}(j,k) = a_{ijk}$, $B_{j}(i,k) = a_{ijk}$, and
$C_{k}(i,j) = a_{ijk}$ to be all the slices of $\mathcal{A}$. Decide if the
following set of equations:
\begin{equation}
\label{tensorquadeq}
\begin{cases}
\mathbf{v}^{\top} A_{i} \mathbf{w} = 0, & i = 1,\dots,l;\\
\mathbf{u}^{\top} B_{j} \mathbf{w} = 0, & j = 1,\dots,m;\\
\mathbf{u}^{\top} C_{k} \mathbf{v} = 0, & k = 1,\dots,n;
\end{cases}
\end{equation}
has a solution $\mathbf{u}\in\mathbb{F}^{l}$, $\mathbf{v}\in\mathbb{F}^{m}$, $\mathbf{w}\in\mathbb{F}^{n}$, with all $\mathbf{u},\mathbf{v},\mathbf{w}$ nonzero.
\end{problem}

Multilinear systems of equations have been studied since the early 19th
century. For instance, the following result was known more than 150 years ago \cite{Cay}.

\begin{example}
[$2\times2\times2$ hyperdeterminant]\label{neq2Hyperdet} For $\mathcal{A}%
=\llbracket a_{ijk}\rrbracket\in\mathbb{C}^{2\times2\times2}$, define
\begin{multline*}
\operatorname*{Det}\nolimits_{2,2,2}(\mathcal{A}):=\frac{1}{4}\biggl[\det
\left(
\begin{bmatrix}
a_{000} & a_{010}\\
a_{001} & a_{011}%
\end{bmatrix}
+%
\begin{bmatrix}
a_{100} & a_{110}\\
a_{101} & a_{111}%
\end{bmatrix}
\right)  -\det\left(
\begin{bmatrix}
a_{000} & a_{010}\\
a_{001} & a_{011}%
\end{bmatrix}
-%
\begin{bmatrix}
a_{100} & a_{110}\\
a_{101} & a_{111}%
\end{bmatrix}
\right)  \biggr]^{2}\\
-4\det%
\begin{bmatrix}
a_{000} & a_{010}\\
a_{001} & a_{011}
\end{bmatrix}
\det
\begin{bmatrix}
a_{100} & a_{110}\\
a_{101} & a_{111}
\end{bmatrix}.
\end{multline*}
Given a matrix $A \in\mathbb{C}^{n \times n}$, the pair of linear equations $\mathbf{x}^{\top}
A =\mathbf{0}$, $A\mathbf{y} = \mathbf{0}$ has a nontrivial solution ($\mathbf{x}, \mathbf{y}$
both nonzero) if and only if $\det(A) = 0$. Cayley proved a multilinear version that parallels the
matrix case.  The following system of bilinear equations:{\small
\begin{align*}
a_{000}x_{0}y_{0}+a_{010}x_{0}y_{1}+a_{100}x_{1}y_{0}+a_{110}x_{1}y_{1}  &
=0, & a_{001}x_{0}y_{0}+a_{011}x_{0}y_{1}+a_{101}x_{1}y_{0}+a_{111}x_{1}y_{1}
&  =0,\\
a_{000}x_{0}z_{0}+a_{001}x_{0}z_{1}+a_{100}x_{1}z_{0}+a_{101}x_{1}z_{1}  &
=0, & a_{010}x_{0}z_{0}+a_{011}x_{0}z_{1}+a_{110}x_{1}z_{0}+a_{111}x_{1}z_{1}
&  =0,\\
a_{000}y_{0}z_{0}+a_{001}y_{0}z_{1}+a_{010}y_{1}z_{0}+a_{011}y_{1}z_{1}  &
=0, & a_{100}y_{0}z_{0}+a_{101}y_{0}z_{1}+a_{110}y_{1}z_{0}+a_{111}y_{1}z_{1}
&  =0,
\end{align*}}%
has a nontrivial solution ($\mathbf{x}, \mathbf{y}, \mathbf{z} \in \mathbb{C}^2$ all
nonzero) if and only if $\operatorname*{Det}_{2,2,2}(\mathcal{A})=0$. \qed
\end{example}

A remarkable result established in \cite{GKZ2,GKZ1} is that hyperdeterminants generalize to tensors of arbitrary orders, provided that certain dimension restrictions \eqref{eq:gkz} are satisfied. 
We do not formally define the hyperdeterminant\footnote{Roughly speaking, the hyperdeterminant is a polynomial that defines the set of all tangent hyperplanes to the set of rank-$1$ tensors in $\mathbb{C}^{l \times m \times n}$. Gelfand, Kapranov, and Zelevinsky showed  that this set is a hypersurface (i.e., defined by the vanishing of a single polynomial) if and only if the condition \eqref{eq:gkz} is satisfied. Also, to be mathematically precise, these sets lie in projective space.} here; however, it suffices to know that $\operatorname*{Det}_{l,m,n}$ is a homogeneous polynomial with integer coefficients  in the variables $x_{ijk}$ where $i = 1,\dots,l$, $j = 1,\dots, m$, and $k =1,\dots,n$.  Such a polynomial defines a function $\operatorname*{Det}_{l,m,n} : \mathbb{C}^{l \times m \times n} \to \mathbb{C}$ by evaluation at $\mathcal{A} = \llbracket  a_{ijk} \rrbracket  \in \mathbb{C}^{l \times m \times n}$, i.e., setting $x_{ijk} = a_{ijk}$.
 The following generalizes Example~\ref{neq2Hyperdet}.  

\begin{theorem}[Gelfand--Kapranov--Zelevinsky]\label{thm:gkz}
Given a tensor $\mathcal{A} \in \mathbb{C}^{l \times m \times n}$, the hyperdeterminant $\operatorname*{Det}_{l,m,n}$ is defined if and only if $l,m,n$ satisfy:
\begin{equation}\label{eq:gkz}
l \le m+n - 1, \quad m \le l + n - 1, \quad n \le l + m - 1.
\end{equation}
In particular, hyperdeterminants exist when $l = m = n$.  Given any $\mathcal{A} = \llbracket  a_{ijk} \rrbracket \in\mathbb{C}^{l \times m \times n}$ with \eqref{eq:gkz} satisfied, the system
\begin{equation}\label{GKZbilinear}
\begin{split}
\sum\nolimits_{j,k=1}^{m,n}a_{ijk}v_{j}w_{k} &=  0\quad  i = 1,\dots, l; \\
\sum\nolimits_{i,k=1}^{l,n}a_{ijk}u_{i}w_{k} &=  0,\quad  j = 1,\dots,m; \\
\sum\nolimits_{i,j=1}^{l,m}a_{ijk}u_{i}v_{j}   &=  0,\quad  k = 1,\dots,n;
\end{split}
\end{equation}
has a nontrivial complex solution if and only if $\operatorname*{Det}\nolimits_{l,m,n}(\mathcal{A}) = 0$.
\end{theorem}
\begin{remark}
Condition \eqref{eq:gkz} is the $3$-tensor equivalent of ``$m \le n$ and $n \le m$" for the existence of determinants of matrices.
\end{remark}

We shall also examine the following closely related problem. Such systems of bilinear equations have  appeared in other contexts \cite{CohenTomasi97}.
\begin{problem}
[Triple Bilinear Feasibility]\label{triquadfeas} Let $\mathbb{F} = \mathbb{R}$ or $\mathbb{C}$. 
Let $A_{k}, B_{k}, C_{k}\in\mathbb{Q}^{n\times n}$ for $k=1,\dots,n$. Decide if the following set of equations:
\begin{equation}
\begin{cases}
\mathbf{v}^{\top}A_{i}\mathbf{w}=0, & i=1,\dots,n;\\
\mathbf{u}^{\top}B_{j}\mathbf{w}=0, & j=1,\dots,n;\\
\mathbf{u}^{\top}C_{k}\mathbf{v}=0, & k=1,\dots,n;
\end{cases}
\label{tquadeq}
\end{equation}
has a solution $\mathbf{u}, \mathbf{v}, \mathbf{w}\in\mathbb{F}^{n}$, with all $\mathbf{u},\mathbf{v},\mathbf{w}$ nonzero.
\end{problem}

One difference between Problem~\ref{triquadfeas} and Problem~\ref{tripletensorfeas} is that  coefficient matrices $A_i, B_j, C_k$ in \eqref{tquadeq} are allowed to be arbitrary rather than slices of a tensor $\mathcal{A}$ as in \eqref{tensorquadeq}. Furthermore, we always assume $l=m=n$ in Problem~\ref{triquadfeas} whereas Problem~\ref{tripletensorfeas} has no such requirement.

If one could show that Problem~\ref{triquadfeas} is NP-hard for $A_{i}, B_{j}, C_{k}$ arising from $\mathcal{A} \in \mathbb{C}^{n \times n \times n}$ or that Problem~\ref{tripletensorfeas}  is NP-hard on the subset of problems where $\mathcal{A} \in \mathbb{C}^{l \times m \times n}$ has $l,m,n$ satisfying \eqref{eq:gkz}, then deciding whether the bilinear system \eqref{GKZbilinear} has a nonzero solution would be NP-hard. It would follow that deciding whether the hyperdeterminant  is zero is also NP-hard. Unfortunately, our proofs do not achieve either of these. The hardness of the hyperdeterminant is therefore still open (Conjecture~\ref{hyp_conj}). 

Before proving Theorem~\ref{zerosingvaluethm}, we first verify that, as in the case of
quadratic feasibility, it is enough to show NP-hardness of the problem over $\mathbb{C}$.

\begin{lemma}
\label{tensorfeascomplex} Let $\mathcal{A}\in\mathbb{R}^{l\times m\times n}$.
There is a tensor $\mathcal{B}\in\mathbb{R}^{2l\times2m\times2n}$ such that
tensor bilinear feasibility over $\mathbb{R}$ for $\mathcal{B}$ is the same
as tensor bilinear feasibility over $\mathbb{C}$ for $\mathcal{A}$.
\end{lemma}

\begin{proof}
Let $A_i = [a_{ijk}]_{j,k=1}^{m,n}$  for $i = 1,\dots,l$.  Consider the tensor $\mathcal{B} = \llbracket b_{ijk} \rrbracket \in
\mathbb{R}^{2l \times2m \times2n}$ given by setting its slices $B_{i}(j,k) =
b_{ijk}$ as follows:
\[
B_{i} =
\begin{bmatrix}
A_{i} & 0\\
0 & -A_{i}%
\end{bmatrix}
, \ B_{l+i} =
\begin{bmatrix}
0 & A_{i}\\
A_{i} & 0
\end{bmatrix}
, \quad i = 1,\dots,l.
\]
It is straightforward to check that nonzero real solutions to \eqref{tensorquadeq} for
the tensor $\mathcal{B}$ correspond in a one-to-one manner with nonzero complex
solutions to \eqref{tensorquadeq} for $\mathcal{A}$.
\end{proof}

We now come to the proof of the main theorem of this section.  For the argument, we shall need the following elementary fact of linear algebra (easily proved by induction on the number of equations):  a system of $m$ homogeneous linear equations in $n$ unknowns with $m < n$ has at least  one nonzero solution.

\begin{theorem}\label{zerosingvaluethm}
Let $\mathbb{F} = \mathbb{R}$ or $\mathbb{C}$.  Graph $3$-colorability is polynomially reducible to
Problem~\ref{tripletensorfeas} (tensor bilinear feasibility). Thus, Problem~\ref{tripletensorfeas} is NP-hard.
\end{theorem}

\begin{proof}
[of Theorem~\ref{zerosingvaluethm}]Given a graph $G = (V,E)$ with $v = \lvert V\rvert$, we shall form a tensor $\mathcal{A} = \mathcal{A}_G \in\mathbb{Z}^{l
\times m \times n}$ with $l = v(2v+5)$ and $m = n = (2v+1)$ having the
property that system \eqref{tensorquadeq} has a nonzero complex solution if
and only if $G$ has a proper $3$-coloring. 

Consider the following vectors of unknowns:
\[
\mathbf{v} = [x_{1},\dots,x_{v},y_{1}, \dots,y_{v},z]^{\top} \quad\text{and}\quad \mathbf{w} = [\hat{x}_{1},\dots,\hat{x}_{v},\hat{y}_{1},\dots, \hat
{y}_{v},\hat{z}]^{\top}.
\]
The $2 \times2$ minors of the
matrix formed by placing $\mathbf{v}$ and $\mathbf{w}$ side-by-side are
$v(2v+1)$ quadratics, $\mathbf{v}^{\top}A_{i} \mathbf{w}$, $i = 1,\dots,v(2v+1)$, for matrices
$A_{i} \in\mathbb{Z}^{(2v+1) \times(2v +1)}$ with entries in $\{-1,0,1\}$. By
construction, these polynomials have a common nontrivial zero $\mathbf{v},\mathbf{w}
$ if and only if there is $c \in\mathbb{C }$ such that $\mathbf{v}
= c \mathbf{w}$.  Next, we write down the $3v$ polynomials $\mathbf{v}^{\top}A_{i} \mathbf{w}$
for $i = v(2v+1)+1,\dots,v(2v+1)+3v$ whose vanishing (along with the equations
above) implies that the $x_{i}$ are cube roots of unity; see
\eqref{CGdef}. We also encode $v$ equations $\mathbf{v}^{\top}A_{i}
\mathbf{w}$ for $i = v(2v+4)+1,\dots,v(2v+4)+v$ whose vanishing implies that
$x_{i}$ and $x_{j}$ are different if $\{i,j\} \in E $. Finally, $\mathcal{A}_G =
\llbracket a_{ijk} \rrbracket \in\mathbb{Z}^{l \times m \times n}$ is defined by
$a_{ijk} = A_{i}(j,k).$

We verify that $\mathcal{A}$ has the claimed property. Suppose that there are
three nonzero complex vectors $\mathbf{u},\mathbf{v},\mathbf{w}$ which satisfy
tensor bilinear feasibility. Then from construction,
$\mathbf{v} = c \mathbf{w}$ for some $c \neq 0$, and also $\mathbf{v}$ encodes a proper
$3$-coloring of the graph $G$ by Lemma~\ref{colorencodinglemma}. Conversely, suppose that $G$ is
$3$-colorable with a coloring represented using a vector $[x_{1}
,\dots,x_{v}]^{\top} \in\mathbb{C}^{v}$ of cube roots of unity.
Then, the vectors $\mathbf{v} = \mathbf{w} = [x_{1},\dots,x_{v},x_{1}
^{-1},\dots,x_{v}^{-1},1]^{\top}$ satisfy the first set of equations in
\eqref{tensorquadeq}. The other sets of equations define a homogeneous linear system for
the vector $\mathbf{u}$ consisting of $4v+2$ equations in $l = v(2v+5) > 4v
+2$ unknowns. In particular, there is always a nonzero $\mathbf{u}$ solving
them, proving that tensor bilinear feasibility is true for
$\mathcal{A}$.
\end{proof}
Note that the $l,m,n$ in this construction do not satisfy \eqref{eq:gkz}; thus,  NP-hardness of the hyperdeterminant does not follow from Theorem~\ref{zerosingvaluethm}. We now prove the following.

\begin{theorem}
\label{tri_quad_feas_thm}Problem~\ref{triquadfeas}, triple bilinear feasibility, is NP-hard over $\mathbb{R}$.
\end{theorem}

\begin{proof}
Since the encoding in Theorem~\ref{QFHard} has more equations than unknowns, we may use Lemma~\ref{quadfeasmoreeqvars} to further transform this system into an equivalent one that is square (see Example~\ref{3OL_ex}).  Thus, if we could solve square quadratic feasibility ($m = n$ in Problem~\ref{HomQuadFeasDef}) over $\mathbb{R}$ in polynomial time, then we could do the same for graph $3$-colorability.  Using this observation, it is enough to prove that a given square, quadratic feasibility problem can be polynomially reduced to Problem~\ref{triquadfeas}. 

Therefore, suppose that $A_{i}$ are given $n \times n$ matrices for which we would
like to determine if $\mathbf{x}^{\top} A_{i} \mathbf{x} = 0$ ($i = 1,\dots
,n$) has a solution $0 \neq\mathbf{x} \in\mathbb{R}^{n}$. Let $E_{ij}$ denote
the matrix with a $1$ in the $(i,j)$ entry and $0$'s elsewhere. Consider a
system $S$ as in \eqref{tquadeq} in which 
\[
B_{1} = C_{1} = E_{11} \quad\text{and}\quad B_{i} = C_{i} = E_{1i} - E_{i1},\quad
\text{for } i=2,\dots,n.
\]
Consider also changing system $S$ by replacing $B_{1}$ and $C_{1}$ with matrices
consisting of all zeroes, and call this system $S^{\prime}$.

We shall construct a decision tree based on answers to feasibility
questions involving systems having form $S$ or $S'$. This will give us an
algorithm to determine whether the original quadratic problem is feasible.  
We make two claims about solutions to $S$, $S^{\prime}$.

\underline{Claim 1}: If $S$ has a solution, then $u_1 = v_1 = w_{1} = 0$.

First note that $u_1v_1 = 0$ since $S$ has a solution.  Suppose that $u_{1} = 0$ and $v_{1}
\neq0$. The form of the matrices $C_{i}$ forces $u_{2} = \dots= u_{n} =
0$. But then $\mathbf{u} = 0$, which contradicts $S$ having a
solution. A similar examination with $u_{1} \neq0$ and $v_{1} = 0$ proves that $u_1 = v_1 = 0$. 
It is now easy to see that we must also have $w_{1} = 0$.

\underline{Claim 2}: Suppose that $S$ has no solution. If $S^{\prime}$ has a
solution, then $\mathbf{v} = c \mathbf{u}$ and $\mathbf{w} = d \mathbf{u}$ for
some $0 \neq c,d \in\mathbb{R}$. Moreover, if $S^{\prime}$ has no solution,
then the original quadratic problem has no solution.

To verify Claim 2, suppose first that $S^{\prime}$ has a solution $\mathbf{u}$, 
$\mathbf{v}$, and $\mathbf{w}$, but $S$ does not. In that case we must have
$u_{1} \neq0$. Also, $v_{1} \neq0$ since otherwise the third set of equations
$\{u_{1} v_{i} - u_{i} v_{1} = 0\}_{i=2}^{n} $ would force $\mathbf{v} = 0$.
But then $\mathbf{v} = c \mathbf{u}$ for $c = \frac{v_{1}}{u_{1}}$ and
$\mathbf{w} = d \mathbf{u}$ for $d = \frac{w_{1}} {u_{1}}$ as desired. On the
other hand, suppose that both $S $ and $S^{\prime}$ have no solution. We claim
that $\mathbf{x}^{\top}A_{i} \mathbf{x} = 0$ ($i = 1,\dots,n$) has no solution
$\mathbf{x} \neq0$ either. Indeed, if it did, then setting $\mathbf{u} =
\mathbf{v} = \mathbf{w} = \mathbf{x}$, we would get a solution to $S^{\prime}$, a contradiction.

We are now prepared to give our method for solving quadratic feasibility using
at most $n+2$ queries to the restricted version ($B_{i} = C_{i}$ for all $i$)
of Problem~\ref{triquadfeas}.

First check if $S$ has a solution. If it does not, then ask if $S^{\prime}$
has a solution. If it does not, then output ``\textsc{no}". This
answer is correct by Claim 2. If $S$ has no solution but $S^{\prime}$ does,
then there is a solution with $\mathbf{v} = c \mathbf{u}$ and $\mathbf{w} = d
\mathbf{u}$, both $c$ and $d$ nonzero. But then $\mathbf{x}^{\top} A_{i}
\mathbf{x} = 0$ for $\mathbf{x} = \mathbf{u}$ and each $i$. Thus, we output
``\textsc{yes}".

If instead, $S$ has a solution, then the solution necessarily has
$(u_{1},v_{1},w_{1}) = (0,0,0)$. Consider now the $n-1$-dimensional system $T$
in which $A_{i}$ becomes the lower-right $(n-1) \times(n-1)$ block of $A_{i}$,
and $C_{i}$ and $D_{i}$ are again of the same form as the previous ones. This
is a smaller system with one less unknown. We now repeat the above
examination inductively with start system $T$ replacing $S$.

If we make it to the final stage of this process without outputting an answer,
then the original system $S$ has a solution with
\[
u_{1} = \dots= u_{n-1} = v_{1} = \dots= v_{n-1} = w_{1} = \dots= w_{n-1} =
0 \text{ and } u_{n},v_{n},w_{n} \text{ are all nonzero}.
\]
It follows that the $(n,n)$ entry of each $A_{i}$ ($i = 1\dots,n$) is zero, and thus 
there is a nonzero solution $\mathbf{x}$ to the the original
quadratic feasibility problem; so we output ``\textsc{yes}".

We have therefore verified the algorithm terminates with the correct answer
and it does so in polynomial time using an oracle that can solve Problem
\ref{triquadfeas}.
\end{proof}
Although in the above proof, we have $l = m =n$ and thus \eqref{eq:gkz} is satisfied, our choice of the coefficient matrices $A_k, B_k, C_k$ do not arise from slices of a single tensor $\mathcal{A}$. So again, NP-hardness of deciding the vanishing of the hyperdeterminant does not follow.

\section{Combinatorial hyperdeterminant is NP-, \#P-, and VNP-hard}\label{sec:combdet}

There is another notion of hyperdetermnant, which we shall call the \textit{combinatorial hyperdeterminant}\footnote{The hyperdeterminants discussed earlier in Section~\ref{bilineareqs} are then called \textit{geometric hyperdeterminants} for distinction \cite{hla}; these are denoted  $\operatorname{Det}$. The combinatorial determinants are also called \textit{Pascal determinants} by some authors.} and denote by the all lowercase $\det_n$. This quantity (only nonzero for even values of $d$) is defined for tensors $\mathcal{A}=[a_{i_{1}i_{2}\cdots i_{d}}] \in \mathbb{C}^{n \times n \times \dots \times n}$ by the formula:
\[
\det\nolimits_n(\mathcal{A})=\sum\nolimits_{\pi_{2},\dots,\pi_{d}\in \mathfrak{S}_{n}}\operatorname{sgn}
(\pi_{2}\cdots\pi_{d})\prod_{i=1}^{n}a_{i\pi_{2}(i)\cdots\pi_{d}(i)}.
\]
For $d=2$, this definition reduces to the usual expression for the determinant of an $n \times n$ matrix. For such hyperdeterminants, we have the following hardness results from \cite[Corollary~5.5.2]{Barv95} and \cite{Gurvits}.
\begin{theorem}[Barvinok]\label{thm:bar}Let $\mathcal{A} \in \mathbb{Z}^{n \times n \times n \times n}$. Deciding if $\det_n (\mathcal{A}) = 0$ is NP-hard.
\end{theorem}
Barvinok proved Theorem~\ref{thm:bar} by showing that  any directed graph $G$ may be encoded as a $4$-tensor $\mathcal{A}_G$ with integer entries in such a way that
the number of Hamiltonian paths between two vertices is  $\det_n (\mathcal{A}_G)$. The  \#P-hardness follows immediately since enuermating Hamiltonian paths is a well-known \#P-complete problem \cite{Valiant79}.

\begin{theorem}[Gurvits]\label{thm:gur}Let $\mathcal{A} \in \{0,1\}^{n \times n \times n \times n}$. Computing $\det_n (\mathcal{A})$ is \#P-hard.
\end{theorem}
Theorem~\ref{thm:gur} is proved by showing  that one may express the permament in terms of the combinatorial hyperdeterminant; the required \#P-hardness then follows from the  
\#P-completeness of the permanent \cite{Valiant}. Even though VNP-hardness was not discussed in \cite{Gurvits}, one may deduce the following result from the same argument and the VNP-completeness of \cite{Valiant79}.
\begin{corollary}\label{cor:gur}The homogeneous polynomial $\det_n$ is VNP-hard to compute.
\end{corollary}

\section{Tensor eigenvalue is NP-hard}\label{sec:EVP}

The eigenvalues and eigenvectors of a symmetric matrix $A\in\mathbb{R}
^{n\times n}$ are the stationary values and points of its Rayleigh
quotient $\mathbf{x}^{\top}A\mathbf{x/x}^{\top}\mathbf{x}$. Equivalently, one may consider the problem of maximizing the quadratic form $\mathbf{x}^{\top}A\mathbf{x}$ constrained to the unit $\ell^2$-sphere:
\begin{equation}\label{l2sphere}
\lVert\mathbf{x}\rVert_{2}^{2}=x_{1}^{2}+x_2^2 + \dots+x_{n}^{2}=1,
\end{equation}
which has the associated Lagrangian, $L(\mathbf{x},\lambda
)=\mathbf{x}^{\top}A\mathbf{x}-\lambda(\lVert\mathbf{x}\rVert_{2}^{2}-1)$. The first order condition, also known as the Karush--Kuhn--Tucker (KKT) condition, at a
stationary point $(\lambda,\mathbf{x})$ yields the familiar eigenvalue
equation $A\mathbf{x}=\lambda\mathbf{x}$, which is then used to define
eigenvalue/eigenvector pairs for any square matrices.

The above discussion extends to give a notion of eigenvalues and
eigenvectors for $3$-tensors. They are suitably constrained stationary values and points of
the cubic form:
\begin{equation}\label{eq:cubic}
\mathcal{A}(\mathbf{x},\mathbf{x},\mathbf{x}) :=\sum_{i,j,k=1}
^{n}a_{ijk}x_{i}x_{j}x_{k},
\end{equation}
associated with a tensor $\mathcal{A}\in\mathbb{R}^{n\times
n\times n}$. However, one now has several natural generalizations of the constraint.  One may
retain \eqref{l2sphere}. Alternatively, one may choose
\begin{equation}\label{l3sphere}
\lVert\mathbf{x}\rVert_{3}^{3}=\lvert x_{1}\rvert^{3}+\lvert x_{2}\rvert
^{3}+\dots+\lvert x_{n}\rvert^{3}=1
\end{equation}
or a unit sum-of-cubes,
\begin{equation}\label{cubesum}
x_{1}^{3}+x_{2}^{3}+\dots+x_{n}^{3}=1. 
\end{equation}
Each choice has an advantage: condition \eqref{l3sphere} defines a
compact set while condition \eqref{cubesum} defines an algebraic set, and both
result in scale-invariant eigenvectors.  Condition
\eqref{l2sphere} defines a set that is both compact and algebraic, but produces
eigenvectors that are not scale-invariant. These were proposed
independently in \cite{L2,Qi}.

By considering the stationarity conditions of the Lagrangian,
$L(\mathbf{x},\lambda)=\mathcal{A}(\mathbf{x},\mathbf{x},\mathbf{x})-\lambda
c(\mathbf{x})$, for $c(\mathbf{x})$ defined by the conditions in \eqref{l2sphere},
\eqref{l3sphere}, or \eqref{cubesum}, we obtain the following.

\begin{definition}\label{def:Eigl2}
Fix $\mathbb{F} = \mathbb{R}$ or $\mathbb{C}$. The number $\lambda\in\mathbb{F}$ is called an $\ell^{2}$-\textbf{eigenvalue} of the tensor $\mathcal{A}\in\mathbb{F}^{n\times n\times n}$ and $\mathbf{0} \neq \mathbf{x}\in\mathbb{F}^{n}$ its corresponding $\ell^{2}$-\textbf{eigenvector} if \eqref{l2eig} holds.  Similarly, $\lambda\in\mathbb{F}$ is an $\ell^3$-\textbf{eigenvalue} and $\mathbf{0} \neq \mathbf{x}\in\mathbb{F}^{n}$ its  $\ell^3$-\textbf{eigenvector} if
\begin{equation}
\sum_{i,j=1}^{n}a_{ijk}x_{i}x_{j}=\lambda x_{k}^{2},\quad
k=1,\dots,n. \label{l3eig}
\end{equation}
\end{definition}

Using the tools we have developed, we prove that real tensor eigenvalue is NP-hard.

\begin{proof}[of Theorem~\ref{TEvalueNP}]
The case $\lambda = 0$ of tensor $\lambda$-eigenvalue becomes square quadratic feasibility ($m=n$ in Problem~\ref{HomQuadFeasDef}) as discussed in the proof of Theorem~\ref{tri_quad_feas_thm}.  Thus, deciding if $\lambda = 0$ is an eigenvalue of a tensor is NP-hard over $\mathbb R$ by Theorem~\ref{TEvalueNP}.   A similar situation holds when we use \eqref{l3eig} to define $\ell^3$-eigenpairs.  
\end{proof}


\begin{figure}[!t]
\begin{center}
\includegraphics[width=3in]{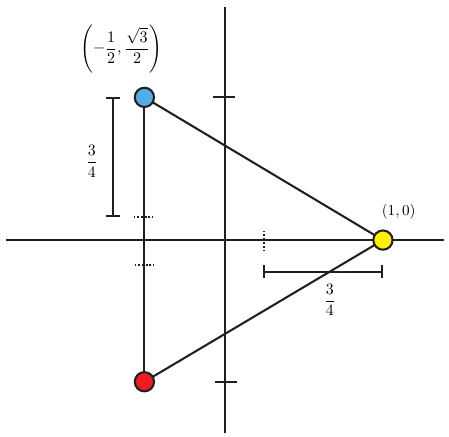}
\caption{\textbf{It is NP-hard to approximate a real eigenvector.} Each colored circle above in the complex plane represents a pair of real numbers which are coordinates of a cube root of unity.  If one could approximate an eigenvector of a rational tensor  to within $\varepsilon = \frac{3}{4}$ in each real coordinate, then one would be able to properly color the vertices of a $3$-colorable graph $G$ (see Example~\ref{3OL_ex}).}\label{approx_cube_roots}
\end{center}\end{figure}

We will see in Section~\ref{symm_tensor_eig} that the eigenvalue problem for \textit{symmetric} $3$-tensors is also NP-hard. We close this section with a proof  that it is even NP-hard to approximate an eigenvector of a tensor.

\begin{proof}[of Theorem~\ref{approx_evect}]
Suppose that one could approximate in polynomial time a tensor eigenvector with eigenvalue $\lambda = 0$ to within $\varepsilon = \frac{3}{4}$ as in \eqref{approx_def}.   By the discussion in Section~\ref{quadfeas}, given a graph $G$, we can form a square set of polynomial equations (see Example~\ref{3OL_ex}), having eigenvectors of a rational tensor $\mathcal{A}_G$ as solutions, encoding proper $3$-colorings of $G$.   Since such vectors represent cube roots of unity separated by a distance of at least $\frac{3}{2}$ in each real or imaginary part (see Fig.~\ref{approx_cube_roots}), finding an approximate real eigenvector to within $\varepsilon = \frac{3}{4}$ of an actual one in polynomial time would allow one to also decide graph $3$-colorability in polynomial time.
\end{proof}

\section{Tensor singular value and spectral norm are NP-hard}\label{sec:SVP}

It is easy to verify that the singular values and singular vectors of a matrix
$A\in\mathbb{R}^{m\times n}$ are the stationary values and stationary points
of the quotient $\mathbf{x}^{\top}A\mathbf{y}/\mathbf{\lVert\mathbf{x}\rVert
}_{2}\mathbf{\lVert\mathbf{y}\rVert}_{2}$. Indeed, the associated
Lagrangian is
\begin{equation}\label{eq:lagrange2}
L(\mathbf{x},\mathbf{y},\sigma)=\mathbf{x}^{\top}A\mathbf{y}-\sigma
(\lVert\mathbf{x}\rVert_{2}\lVert\mathbf{y}\rVert_{2}-1),
\end{equation}
and the first order condition yields, at a stationary point $(\mathbf{x}, \mathbf{y})$, the familiar singular value equations:
\[
A\mathbf{v}=\sigma\mathbf{u},\quad A^{\top}\mathbf{u}=\sigma\mathbf{v},
\]
where $\mathbf{u}=\mathbf{x}/\lVert\mathbf{x}\rVert_{2}$ and $\mathbf{v}
=\mathbf{y}/\lVert\mathbf{y}\rVert_{2}$.

This derivation has been extended to define singular values and singular vectors for higher-order tensors \cite{L2}. For
$\mathcal{A}\in\mathbb{R}^{l\times m\times n}$, we have the trilinear form\footnote{When $l=m=n$ and $\mathbf{x} = \mathbf{y} = \mathbf{z}$,
the trilinear form in \eqref{eq:trilin} becomes the cubic form in \eqref{eq:cubic}.}:
\begin{equation}\label{eq:trilin}
\mathcal{A}(\mathbf{x},\mathbf{y},\mathbf{z}) :=\sum_{i,j,k=1}
^{n}a_{ijk}x_{i}y_{j}z_{k},
\end{equation}
and consideration of its stationary values on a product of unit $\ell^p$-spheres leads to the Lagrangian,
\[
L(\mathbf{x},\mathbf{y},\mathbf{z},\sigma)=\mathcal{A}(\mathbf{x}
,\mathbf{y},\mathbf{z})-\sigma(\lVert\mathbf{x}\rVert_{p}\lVert\mathbf{y}
\rVert_{p}\lVert\mathbf{z}\rVert_{p}-1).
\]
The only ambiguity is choice of $p$. As for eigenvalues, natural choices are $p=2$ or $3$.
\begin{definition}\label{def:Sing}
Fix $\mathbb{F} = \mathbb{R}$ or $\mathbb{C}$.  Let $\sigma\in\mathbb{F}$, and suppose that $\mathbf{u} \in\mathbb{F}
^{l}$, $\mathbf{v}\in\mathbb{F}^{m}$, and $\mathbf{w}\in\mathbb{F}^{n}$ are all nonzero.
The number $\sigma \in \mathbb{F}$ is called an $\ell^2$-\textbf{singular value} and the nonzero $\mathbf{u}
,\mathbf{v},\mathbf{w}$ are called $\ell^2$-\textbf{singular vectors} of $\mathcal{A}$ if
\begin{equation}\label{l2sing}
\begin{split}
\sum\nolimits_{j,k=1}^{m,n}a_{ijk}v_{j}w_{k} &= \sigma u_{i},\quad  i = 1,\dots, l; \\
\sum\nolimits_{i,k=1}^{l,n}a_{ijk}u_{i}w_{k} &= \sigma v_{j},\quad  j = 1,\dots,m;\\
\sum\nolimits_{i,j=1}^{l,m}a_{ijk}u_{i}v_{j}   &= \sigma w_{k},\quad  k = 1,\dots,n.
\end{split}
\end{equation}
Similarly,  $\sigma$ is called an $\ell^3$-\textbf{singular value} and nonzero $\mathbf{u},\mathbf{v},\mathbf{w}$
$\ell^3$-\textbf{singular vectors} if 
\begin{equation}\label{l3sing}
\begin{split}
\sum\nolimits_{j,k=1}^{m,n}a_{ijk}v_{j}w_{k} &=  \sigma u_{i}^2,\quad  i = 1,\dots, l; \\
\sum\nolimits_{i,k=1}^{l,n}a_{ijk}u_{i}w_{k} &=  \sigma v_{j}^2,\quad  j = 1,\dots,m; \\
\sum\nolimits_{i,j=1}^{l,m}a_{ijk}u_{i}v_{j}   &=  \sigma w_{k}^2,\quad  k = 1,\dots,n. 
\end{split}
\end{equation}
\end{definition}
When $\sigma = 0$, definitions \eqref{l2sing} and \eqref{l3sing} agree and reduce to tensor bilinear feasibility (Problem~\ref{tripletensorfeas}). In particular, if condition \eqref{eq:gkz} holds, then $\operatorname{Det}_{l,m,n}(\mathcal{A}) = 0$ iff $0$ is an $\ell^2$-singular value of $\mathcal{A}$ iff $0$ is an $\ell^3$-singular value of $\mathcal{A}$ \cite{L2}. 

%

The following is immediate from Theorem~\ref{zerosingvaluethm}, which was proved by a reduction from $3$-colorability.

\begin{theorem}\label{singvalue0thm} 
Let $\mathbb{F} = \mathbb{R}$ or $\mathbb{C}$.   Deciding whether $\sigma=0$ is an ($\ell^2$ or $\ell^3$) singular
value over $\mathbb{F}$ of a tensor is NP-hard.
\end{theorem}

%

The tools we  developed in Section~\ref{bilineareqs} also directly apply to give an analogue of Theorem~\ref{approx_evect} for approximating singular vectors corresponding to singular value $\sigma =0$.  

\begin{theorem}\label{approx_singvect}
It is NP-hard to approximate a triple of tensor singular vectors over $\mathbb{R}$ to within $\varepsilon = \frac{3}{4}$ and over $\mathbb{C}$ to within $\varepsilon = \frac{\sqrt{3}}{2}$.
\end{theorem}

\begin{corollary}
Unless $\mathit{P} = \mathit{NP}$, there is no PTAS for approximating tensor singular vectors.
\end{corollary}

Note that verifying whether $0 \neq \sigma \in \mathbb{Q}$ is a singular value of  $\mathcal{A}$ is the same as checking whether $1$ is a singular value of $\mathcal{A}/\sigma$.  In this section, we reduce computing the max-clique number of a graph to a singular value problem for $\sigma = 1$, extending some ideas of \cite{nesterov03}, \cite{HLZ}.  In particular, we shall prove the following.

\begin{theorem}\label{singvalue1thm} 
Fix $0 \neq \sigma \in \mathbb{Q}$.  Deciding whether $\sigma$ is an $\ell^2$-singular
value over $\mathbb{R}$ of a tensor is NP-hard.
\end{theorem}

We next define the closely related concept of spectral norm of a tensor.

\begin{definition}\label{defn:spec_norm}
The \textbf{spectral norm} of a tensor $\mathcal{A}$ is
\[
\lVert\mathcal{A}\rVert_{2,2,2}:=\sup_{\mathbf{x},\mathbf{y},\mathbf{z}
\neq\mathbf{0}}\frac{\lvert\mathcal{A}(\mathbf{x},\mathbf{y},\mathbf{z}
)\rvert}{\lVert\mathbf{x}\rVert_{2}\lVert\mathbf{y}\rVert_{2}\lVert
\mathbf{z}\rVert_{2}}.
\]
\end{definition}

The spectral norm is either the maximum or minimum value of $\mathcal{A}(\mathbf{x}
,\mathbf{y},\mathbf{z})$ constrained to the set $\{(\mathbf{x},\mathbf{y}
,\mathbf{z}) :  \lVert\mathbf{x}\rVert_{2}=\lVert\mathbf{y}\rVert_{2}=
\lVert\mathbf{z}\rVert_{2}=1\}$, and thus is an $\ell^2$-singular value of
$\mathcal{A}$.  At the end of this section, we will show that the corresponding spectral norm questions are NP-hard (Theorems~\ref{spec_norm_nphard_thm} and \ref{approx_spec_nphard_thm}).

We now explain our setup for the proof of Theorem~\ref{singvalue1thm}.
Let $G = (V,E)$ be a simple graph on vertices $V = \{1,\dots,v\}$ with $e$
edges $E$, and let $\omega(G)$ be the \textit{clique number} of $G$ (that is,
the number of vertices in a largest clique). Given a graph $G$ and $l \in \mathbb{N}$, deciding whether $\omega(G) \geq l$ is one of the first decision problems known to be NP-complete \cite{Karp}. An important result
linking an optimization problem to $\omega (G)$ is the following
classical theorem \cite{MotzkinStraus}. It can be used
to give an elegant proof of Tur\'an's Graph Theorem, which bounds the number
of edges in a graph in terms of its clique number (e.g., see \cite{Aigner}).

\begin{theorem}[Motzkin--Straus]\label{Motzkin-Straus} Let $\Delta_{v} := \{(x_{1},\dots,x_{v})
\in\mathbb{R}_{\geq0}^{v}: \sum_{i =1}^{v} x_{i} = 1\}$ and let $G = (V,E)$ be
a graph on $v$ vertices with clique number $\omega(G)$. Then,
\[
1 - \frac{1}{\omega(G)} = 2 \cdot\max_{ \mathbf{x} \in\Delta_{v}} \sum\nolimits_{\{i,j\} \in E}
x_{i} x_{j}.
\]
\end{theorem}

Let $A_{G}$ be the adjacency matrix of the graph $G$ and set $\omega = \omega(G)$. For each positive
integer $l$, define $Q_l := A_{G} + \frac{1}{l} J$, in which $J$ is the all-ones matrix. Also, let
\[
M_l := \max_{ \mathbf{x} \in\Delta_{v}} \mathbf{x}^{\top}Q_l \mathbf{x} = 1 + \frac{\omega- l}{l \omega},
\]
where the second equality follows from Theorem~\ref{Motzkin-Straus}.  We have $M_{\omega} = 1$ and also 
\begin{equation}\label{mlfacts}
M_l > 1 \ \text{ if } \ l < \omega; \quad M_l < 1 \ \text{ if } \ l > \omega.
\end{equation}

For $k=1,\dots,e$, let $E_{k}=\frac{1}{2}E_{i_{k}j_{k}}+\frac{1}{2}
E_{j_{k}i_{k}}$ in which $\{i_{k},j_{k}\}$ is the $k$th edge of $G$. Here, the
$v\times v$ matrix $E_{ij}$ has a 1 in the $(i,j)$-th spot and zeroes
elsewhere. For each positive integer $l$, consider the following
optimization problem (having rational input):
\[
N_l:=\max_{\lVert\mathbf{u}\rVert_{2}=1}\left\{  \sum_{i=1}^l\left(  \mathbf{u}^{\top}\frac
{1}{l}I \mathbf{u}\right)  ^{2}+2 \sum_{k=1}^{e}(\mathbf{u}^{\top}E_{k} \mathbf{u})^{2}\right\}.
\]

\begin{lemma}\label{lem:M=N}
For any graph $G$, we have $M_l = N_l$.
\end{lemma}

\begin{proof}
By construction,  $N_l = \frac{1}{l} + 2 \cdot\max_{\lVert\mathbf{u}\rVert_2 =1} \sum_{\{i,j\} \in E}
u_{i}^{2} u_{j}^{2},$ which is easily seen to equal $M_l$.
\end{proof}

We next state a beautiful result of Banach \cite{B,PST} that will be very useful for us here as well as in Section~\ref{sec:symmproblems}. The result essentially says that the spectral norm of a symmetric tensor may be expressed symmetrically.
\begin{theorem}[Banach]\label{thm:banach1}
Let $\mathcal{S}\in\mathbb{R}^{n\times n \times n}$ be a symmetric $3$-tensor. Then
\begin{equation}
\lVert\mathcal{S}\rVert_{2,2,2} = \sup_{\mathbf{x},\mathbf{y},\mathbf{z}
\neq\mathbf{0}}\frac{\lvert\mathcal{S}(\mathbf{x},\mathbf{y},\mathbf{z}
)\rvert}{\lVert\mathbf{x}\rVert_{2}\lVert\mathbf{y}\rVert_{2}\lVert
\mathbf{z}\rVert_{2}}=\sup_{\mathbf{x}\neq\mathbf{0}}\frac{\lvert
\mathcal{S}(\mathbf{x},\mathbf{x},\mathbf{x})\rvert}{\lVert\mathbf{x}
\rVert_{2}^{3}}.
\label{eq:banach3}
\end{equation}
Let $\mathcal{S}\in\mathbb{R}^{n\times n \times n \times n}$ be a symmetric $4$-tensor. Then
\begin{equation}
\lVert\mathcal{S}\rVert_{2,2,2,2}=\sup_{\mathbf{w}, \mathbf{x},\mathbf{y},\mathbf{z}
\neq\mathbf{0}}\frac{\lvert\mathcal{S}(\mathbf{w}, \mathbf{x},\mathbf{y},\mathbf{z}
)\rvert}{ \lVert \mathbf{w} \rVert_2 \lVert\mathbf{x}\rVert_{2}\lVert\mathbf{y}\rVert_{2}\lVert
\mathbf{z}\rVert_{2}}=\sup_{\mathbf{x}\neq\mathbf{0}}\frac{\lvert
\mathcal{S}(\mathbf{x},\mathbf{x},\mathbf{x},\mathbf{x})\rvert}{\lVert\mathbf{x}
\rVert_{2}^{4}}.
\label{eq:banach4}
\end{equation}
\end{theorem}
While we have restricted ourselves to orders $3$ and $4$ for  simplicity, Banach's result holds for  arbitrary order. Furthermore, $\mathbb{C}$ may replace $\mathbb{R}$ without affecting its validity.

The following interesting fact, which is embedded in the proof of \cite[Proposition 2]{HLZ}, may be easily deduced from Theorem~\ref{thm:banach1}.
\begin{proposition}[He--Li--Zhang] \label{prop:HLZ}
Let $A_{1},\dots,A_{m}\in\mathbb{R}^{n\times n}$ be symmetric.
Then,
\begin{equation}
\max_{\lVert\mathbf{u}\rVert_{2}=\lVert\mathbf{v}\rVert_{2}=1}\sum_{k=1}%
^{m}(\mathbf{u}^{\top}A_{k}\mathbf{v})^{2}=\max_{\lVert\mathbf{v}\rVert_{2}%
=1}\sum_{k=1}^{m}(\mathbf{v}^{\top}A_{k}\mathbf{v})^{2}.\label{eqhlz0}%
\end{equation}
\end{proposition}

\begin{proof}
Define
\[
f(\mathbf{u},\mathbf{v},\mathbf{w},\mathbf{x}) := \sum_{k=1}^{m}(\mathbf{u}^{\top}A_{k}\mathbf{v})(\mathbf{w}^{\top}A_{k}\mathbf{x}).
\]
Clearly, we must have
\[
\max_{\lVert \mathbf{v} \rVert_2 =1}  f(\mathbf{v},\mathbf{v},\mathbf{v},\mathbf{v}) \le
\max_{\lVert \mathbf{u} \rVert_2 =\lVert \mathbf{v} \rVert_2 =1}  f(\mathbf{u},\mathbf{v},\mathbf{u},\mathbf{v}) \le
\max_{\lVert \mathbf{u} \rVert_2 =\lVert \mathbf{v} \rVert_2=\lVert \mathbf{w} \rVert_2=\lVert \mathbf{x} \rVert_2 =1}  f(\mathbf{u},\mathbf{v},\mathbf{w},\mathbf{x}).
\]
Note that we may write $f(\mathbf{u},\mathbf{v},\mathbf{w},\mathbf{x}) = \mathcal{S}(\mathbf{u}, \mathbf{v},\mathbf{w},\mathbf{x})$ for some symmetric $4$-tensor \mbox{$\mathcal{S} \in \mathbb{R}^{n \times n \times n \times n}$.} Since the first and last terms in the inequality above are equal by \eqref{eq:banach4} in Banach's theorem, we obtain \eqref{eqhlz0}. \fin
\end{proof}

\begin{lemma}\label{mainTlem}
The maximization problem
\begin{equation}\label{tensoroptform}
T_l:=\max_{\lVert\mathbf{u}\rVert_2=\lVert\mathbf{v}\rVert_{2}=\lVert\mathbf{w}\rVert_{2}=1} \left\{  \sum_{i=1}^l\left(  \mathbf{u}^{\top}\frac
{1}{l}I \mathbf{v}\right)  w_{i}+\sum_{k=1}^{e}(\mathbf{u}^{\top}E_{k}\mathbf{v})w_{l + k}+\sum
_{k=1}^{e}(\mathbf{u}^{\top}E_{k}\mathbf{v})w_{m+l + k} \right\}
\end{equation}
has optimum value  $T_l = M_l^{1/2}$. Thus,
\[
T_l = 1 \quad\text{iff}\quad l = \omega; \quad T_l > 1 \quad\text{iff}\quad l < \omega; \quad \text{and}\quad T_l < 1 \quad\text{iff}\quad l > \omega,
\]
\end{lemma}

\begin{proof}
Fixing $\mathbf{a}=[a_{1},\dots,a_{s}]^{\top} \in \mathbb{R}^s$, the Cauchy-Schwarz inequality implies that a sum
 $\sum_{i=1}^{s}a_{i}w_{i}$ with $\lVert\mathbf{w}\rVert_{2}=1$ achieves a maximum value of $\lVert\mathbf{a}\rVert_{2}$ (with $w_i = a_i/\lVert\mathbf{a}\rVert_{2}$ if $\lVert\mathbf{a}\rVert_{2} \neq 0$). 
Thus,
\begin{align*}
T_l &=    \max_{\lVert\mathbf{u}\rVert_{2}=\lVert\mathbf{v}\rVert_{2}=\lVert\mathbf{w}\rVert_{2}=1}\sum_{i=1}^l\left(  \mathbf{u}^{\top
}\frac{1}{l}I\mathbf{v}\right)  w_{i}+\sum_{k=1}^{e}(\mathbf{u}^{\top}E_{k}\mathbf{v})w_{l + k}%
+\sum_{k=1}^{e}(\mathbf{u}^{\top}E_{k}\mathbf{v})w_{e+l + k}\\
&=  \max_{\lVert\mathbf{u}\rVert_{2}=\lVert\mathbf{v}\rVert_{2}=1}\sqrt{\sum_{i=1}^l\left( \mathbf{u}^{\top}\frac
{1}{l}I\mathbf{v}\right)  ^{2}+2\sum_{k=1}^{e}(\mathbf{u}^{\top}E_{k}\mathbf{v})^{2}}\\
&=  M_l^{1/2},
\end{align*}
where the last equality follows from Lemma~\ref{lem:M=N} and Proposition~\ref{prop:HLZ}.
\end{proof}

We can now prove Theorem~\ref{singvalue1thm}, and Theorems~\ref{spec_norm_nphard_thm} and \ref{approx_spec_nphard_thm} from the introduction.
\begin{proof}
[of Theorem~\ref{singvalue1thm}]We cast \eqref{tensoroptform} in the
form of a tensor singular value problem. Set $\mathcal{A}_l$ to be the
three dimensional tensor with $a_{ijk}$ equal to the coefficient of the term
$u_{i} v_{j} w_{k}$ in the multilinear form \eqref{tensoroptform}. Then 
$T_l$ is just the maximum $\ell^2$-singular value of $\mathcal{A}_l$.
We now show that if we could decide whether $\sigma = 1$ is an $\ell^2$-singular value
of $\mathcal{A}_l$, then we would solve the max-clique problem. 

Given a graph $G$, construct the tensor $\mathcal{A}_l$ for each integer $l \in\{1,\dots,v\}$.
The largest value of $l$ for which $1$ is a singular value of  $\mathcal{A}_l$ is $\omega = \omega(G)$.
To see this, notice that if $l$ is larger than $\omega$, the maximum singular value of $\mathcal{A}_l$ is smaller
than $1$ by Lemma~\ref{mainTlem}. Therefore, $\sigma= 1$ can not be a singular
value of $\mathcal{A}_l$ in these cases. However, $\sigma=1$ is a
singular value of the tensor $\mathcal{A}_{\omega}$.
\end{proof}

\begin{proof}[of Theorem~\ref{spec_norm_nphard_thm}]
In the reduction above used to prove Theorem~\ref{singvalue1thm}, it suffices to decide which tensor $\mathcal{A}_l$ has spectral norm equal to $1$.
\end{proof}

\begin{proof}[of Theorem~\ref{approx_spec_nphard_thm}]
Suppose that we could approximate spectral norm to within a factor of $1-\varepsilon = (1+1/N(N-1))^{-1/2}$, where $N$ is tensor input size.  Consider the tensors $\mathcal{A}_l$ as in the proof of Theorems~\ref{singvalue1thm} and \ref{spec_norm_nphard_thm} above, which have input size $N$ that is at least the number of vertices $v$ of the graph $G$.  For each $l$, we are guaranteed an approximation for the spectral norm of  $\mathcal{A}_{l}$ of at least 
\begin{equation}\label{apprx_spec_eqn}
(1-\varepsilon)\cdot M_{l}^{1/2}  > \left(1+\frac{1}{v(v-1)}\right)^{-1/2}\left(1+\frac{\omega-l}{l\omega}\right)^{1/2}.
\end{equation}
It is easy to verify that  (\ref{apprx_spec_eqn}) implies that any spectral norm approximation of $\mathcal{A}_{l}$ is greater than $1$ whenever $l \leq \omega-1$.  
In particular, as $\mathcal{A}_{\omega}$ has spectral norm exactly $1$, we can determine $\omega$ by finding the largest $l = 1,\ldots, v$ for which a spectral norm approximation of $\mathcal{A}_l$ is $1$ or less.
\end{proof}

\section{Best rank-$1$  tensor approximation is NP-hard}\label{sec:rank1}

We shall need to define the \textit{Frobenius norm} and \textit{inner product} for this and later sections. Let $\mathcal{A} = \llbracket a_{ijk} \rrbracket_{i,j,k=1}^{l,m,n}$ and $\mathcal{B} = \llbracket b_{ijk} \rrbracket_{i,j,k=1}^{l,m,n}  \in \mathbb{R}^{n \times n \times n}$. Then we define:
\[
\lVert \mathcal{A} \rVert_F^2 := \sum_{i,j,k=1}^{l,m,n} \lvert a_{ijk} \rvert^2, \qquad \langle \mathcal{A}, \mathcal{B} \rangle := \sum_{i,j,k=1}^{l,m,n} a_{ijk} b_{ijk}.
\]
Clearly $\lVert \mathcal{A} \rVert_F^2 = \langle \mathcal{A}, \mathcal{A} \rangle$ and  $\langle \mathcal{A}, \mathbf{x} \otimes \mathbf{y} \otimes \mathbf{z} \rangle = \mathcal{A}( \mathbf{x}, \mathbf{y},\mathbf{z})$, where $\mathcal{A}( \mathbf{x}, \mathbf{y},\mathbf{z})$ is as in \eqref{eq:trilin}.\fin

As we explain next, the best rank-$r$ approximation problem for a tensor is well-defined only when $r = 1$.  The general  problem can be expressed as solving:   
\[
\min_{\mathbf{x}_{i},\mathbf{y}_{i},\mathbf{z}_{i}}\lVert\mathcal{A}
-\lambda_{1}\mathbf{x}_{1}\otimes\mathbf{y}_{1}\otimes\mathbf{z}_{1}
-\dots-\lambda_{r}\mathbf{x}_{r}\otimes\mathbf{y}_{r}\otimes\mathbf{z}
_{r}\rVert_{F}.
\]
Unfortunately, a solution to this optimization problem does not necessarily exist; in fact, the set $\{\mathcal{A}\in\mathbb{R}^{l\times m\times n} : 
\operatorname{rank}_{\mathbb{R}}(\mathcal{A})\leq r\}$ is not closed, in general, when
$r>1$. The following simple example is based on an exercise in \cite{Kn}.

\begin{example}
Let $\mathbf{x}_{i},\mathbf{y}_{i}\in\mathbb{R}^{m}$, $i=1,2,3$. Let
\[
\mathcal{A}=\mathbf{x}_{1}\otimes\mathbf{x}_{2}\otimes\mathbf{y}_{3}+\mathbf{x}
_{1}\otimes\mathbf{y}_{2}\otimes\mathbf{x}_{3}+\mathbf{y}_{1}\otimes
\mathbf{x}_{2}\otimes\mathbf{x}_{3},
\]
and for $n\in\mathbb{N}$, let
\[
\mathcal{A}_{n}=\mathbf{x}_{1}\otimes\mathbf{x}_{2}\otimes(\mathbf{y}_{3}
-n\mathbf{x}_{3})+\left(  \mathbf{x}_{1}+\frac{1}{n}\mathbf{y}_{1}\right)
\otimes\left(  \mathbf{x}_{2}+\frac{1}{n}\mathbf{y}_{2}\right)  \otimes
n\mathbf{x}_{3}.
\]
One can show that $\operatorname{rank}_{\mathbb{R}}(\mathcal{A})=3$ if and only if the pair $\mathbf{x}_{i},\mathbf{y}_{i}$ are linearly independent, $i=1,2,3$. Since $\operatorname{rank}_{\mathbb{F}}(\mathcal{A}_{n})\leq2$ and
\[
\lim_{n\rightarrow\infty} \mathcal{A}_{n}=\mathcal{A},
\]
the rank-$3$ tensor $\mathcal{A}$ has no best rank-$2$ approximation. \qed
\end{example}

The phenomenon of a tensor failing to have a best rank-$r$ approximation is
 widespread, occurring over a  range of
dimensions, orders, and ranks, regardless of the norm (or
Br\`{e}gman divergence) used. These counterexamples occur with positive
probability and sometimes with certainty (in $\mathbb{R}^{2\times2\times
2}$, no tensor of rank-$3$ has a best rank-$2$ approximation). We refer the
reader to \cite{dSL} for further details.

On the other hand, the set of rank-$1$ tensors (together with zero) is  closed.
In fact, it is the Segre variety in classical algebraic
geometry \cite{Landsberg}. Consider the problem of finding the best rank-$1$ approximation to a tensor $\mathcal{A}$:
\begin{equation}\label{eq:bestrank1}
\min_{\mathbf{x},\mathbf{y},\mathbf{z}}\lVert\mathcal{A}-\mathbf{x}%
\otimes\mathbf{y}\otimes\mathbf{z}\rVert_{F}.
\end{equation}
By introducing an additional parameter $\sigma\geq0$, we may rewrite the
rank-$1$ term in the form $\mathbf{x}
\otimes\mathbf{y}\otimes\mathbf{z}=\sigma \mathbf{u}\otimes\mathbf{v}\otimes\mathbf{w}$ where $\lVert\mathbf{u}\rVert_{2}=$ $\lVert
\mathbf{v}\rVert_{2}=\lVert\mathbf{w}\rVert_{2}=1$.  Then,
\begin{align*}
\lVert\mathcal{A}-\sigma\mathbf{u}\otimes\mathbf{v}\otimes\mathbf{w}\rVert
_{F}^{2}  &  =\lVert\mathcal{A}\rVert_{F}^{2}-2\sigma\langle\mathcal{A}%
,\mathbf{u}\otimes\mathbf{v}\otimes\mathbf{w}\rangle+\sigma^{2}\lVert
\mathbf{u}\otimes\mathbf{v}\otimes\mathbf{w}\rVert_{F}^{2}\\
&  =\lVert\mathcal{A}\rVert_{F}^{2}-2\sigma\langle\mathcal{A},\mathbf{u}%
\otimes\mathbf{v}\otimes\mathbf{w}\rangle+\sigma^{2}.
\end{align*}
This expression is minimized when
\[
\sigma=\max_{\lVert\mathbf{u}\rVert_{2}=\lVert\mathbf{v}\rVert_{2}%
=\lVert\mathbf{w}\rVert_{2}=1}\langle\mathcal{A},\mathbf{u}\otimes
\mathbf{v}\otimes\mathbf{w}\rangle = \lVert \mathcal{A} \rVert_{2,2,2}
\]
since $\langle\mathcal{A},\mathbf{u}\otimes\mathbf{v}\otimes\mathbf{w}\rangle
=\mathcal{A}(\mathbf{u},\mathbf{v},\mathbf{w})$.  If $(\mathbf{x},\mathbf{y},\mathbf{z})$ is a solution to the optimization problem \eqref{eq:bestrank1}, then $\sigma$ may be computed as
\[
\sigma = \lVert\sigma \mathbf{u}\otimes\mathbf{v}\otimes\mathbf{w} \rVert_F = \lVert \mathbf{x}
\otimes\mathbf{y}\otimes\mathbf{z} \rVert_F =\lVert \mathbf{x}\rVert_2 \lVert\mathbf{y}\rVert_2 \lVert \mathbf{z} \rVert_2.
\]
We conclude that determining the best rank-$1$ approximation is also NP-hard, which is Theorem~\ref{rank1approx_nphard_thm} from the introduction.  We will see in Section~\ref{sec:symmproblems} that restricting to symmetric $3$-tensors does not make 
the best rank-$1$ approximation problem easier.

\section{Tensor rank is NP-hard}\label{sec:Rank}

It was shown in \cite{Haa} that any 3SAT boolean formula\footnote{Recall that this a boolean formula in $n$ variables and $m$ clauses where each clause contains exactly three variables, e.g.\ $(x_1 \vee \bar{x}_2 \vee \bar{x}_3)\wedge (x_1 \vee x_2 \vee x_4)$.} can be encoded as a $3$-tensor $\mathcal{A}$ over a finite field or $\mathbb{Q}$ and that the satisfiability of the formula is equivalent to checking whether $\operatorname{rank}
(\mathcal{A}) \le r$ for some $r$ that depends on the number of variables and clauses (the tensor rank being taken over the respective field). In particular, tensor rank is NP-hard over $\mathbb{Q}$ and NP-complete over finite fields.

Since the majority of recent applications of tensor methods are over
$\mathbb{R}$ and $\mathbb{C}$, a natural question is whether tensor rank is
also NP-hard over these fields. In other words, is it NP-hard to decide whether $\operatorname{rank}_{\mathbb{R}}(\mathcal{A})\leq
r$ or if $\operatorname{rank}_{\mathbb{C}}(\mathcal{A})\leq r$ for a given tensor $\mathcal{A}$ with rational entries and a given 
$r\in\mathbb{N}$? 

One difficulty with the notion of tensor rank is that it depends on the base field.  For instance, there are real tensors 
with rank over $\mathbb{C}$ strictly less than their rank over $\mathbb{R}$ \cite{dSL}.  
We will show here that the same can happen for tensors with rational entries. In particular,
H\aa stad's result for tensor rank over $\mathbb{Q}$ does not directly apply to $\mathbb{R}$ and $\mathbb{C}$. Nevertheless, H\aa stad's proof shows, as we explain below in Theorem~\ref{rank_nphard_thm}, that tensor rank remains NP-hard over both $\mathbb{R}$ and $\mathbb{C}$.

\begin{proof}[of Theorem~\ref{rankQRthm}]
We explicitly construct a rational tensor $\mathcal{A}$ with $\operatorname{rank}_{\mathbb{R}}(\mathcal{A}) < \operatorname{rank}_{\mathbb{Q}}(\mathcal{A})$. Let $\mathbf{x} = [1,0]^{\top}$ and $\mathbf{y} = [0,1]^{\top}$.  First observe that
\[ \overline{\mathbf{z}}\otimes\mathbf{z}\otimes\overline{\mathbf{z}}
+\mathbf{z}\otimes\overline{\mathbf{z}}\otimes\mathbf{z}=2\mathbf{x}\otimes\mathbf{x}\otimes\mathbf{x}-4\mathbf{y}\otimes\mathbf{y}\otimes\mathbf{x}+4\mathbf{y}\otimes \mathbf{x}\otimes\mathbf{y}-4\mathbf{x}\otimes\mathbf{y}\otimes\mathbf{y} \in \mathbb{Q}^{2 \times 2 \times 2},\]
where $\mathbf{z} =\mathbf{x}+\sqrt{2}\mathbf{y}$ and  $\overline{\mathbf{z}}=\mathbf{x}-\sqrt{2}\mathbf{y}$. Let $\mathcal{A}$ be this tensor; thus, $\operatorname{rank}
\nolimits_{\mathbb{R}}(\mathcal{A})\leq2$. We claim that $\operatorname{rank}
\nolimits_{\mathbb{Q}}(\mathcal{A})>2$. Suppose not and that there exist $\mathbf{u}
_{i}=[a_{i},b_{i}]^{\top}$, $\mathbf{v}_{i}=[c_{i},d_{i}]^{\top}\in
\mathbb{Q}^{2}$, $i=1,2,3$, with
\begin{equation}\label{eq:rank2}
\mathcal{A}=\mathbf{u}_{1}\otimes\mathbf{u}_{2}\otimes\mathbf{u}_{3}+\mathbf{v}
_{1}\otimes\mathbf{v}_{2}\otimes\mathbf{v}_{3}.
\end{equation}
Identity \eqref{eq:rank2} gives eight equations found in \eqref{rationaltensoreqs}.  Thus, by Lemma~\ref{rankdroplemma}, $\operatorname{rank}
\nolimits_{\mathbb{Q}}(\mathcal{A})>2$.
\end{proof}

\begin{lemma}\label{rankdroplemma}
The system of $8$ equations in $12$ unknowns:
\begin{equation}\label{rationaltensoreqs}
\begin{split}
a_{1}a_{2}a_{3}+c_{1}c_{2}c_{3}=\  &  2,\ a_{1}a_{3}b_{2}+c_{1}c_{3}d_{2}=0, \ a_{2}a_{3}b_{1}+c_{2}c_{3}d_{1}=  0, \\
a_{3}b_{1}b_{2}+c_{3}d_{1}d_{2}= \ & -4, \ a_{1}a_{2}b_{3}+c_{1}c_{2}d_{3}= 0,\ a_{1}b_{2}b_{3}+c_{1}d_{2}d_{3}=-4, \\
a_{2}b_{1}b_{3}+c_{2}d_{3}d_{1}=\  &  4,\ b_{1}b_{2}b_{3}+d_{1}d_{2}d_{3}=0
\end{split}
\end{equation}
has no solution in rational numbers $a_{1},a_{2},a_{3}$, $b_{1},b_{2},b_{3}$, $c_{1},c_{2},c_{3}$, and $d_{1},d_{2},d_{3}$.  
\end{lemma}
\begin{proof}
One may verify in exact symbolic arithmetic (see the Appendix) that the following two equations are polynomial consequences of \eqref{rationaltensoreqs}:
\[2c_{2}^{2}-d_{2}^{2} = 0 \quad \text{and} \quad c_{1}d_{2}d_{3}-2 = 0.\]  Since no rational number when squared equals $2$, the first equation implies that any rational solution to \eqref{rationaltensoreqs} must have $c_{2}=d_{2}=0$, an impossibility by the second.  Thus, no rational solutions to \eqref{rationaltensoreqs} exist.
\end{proof}

We now provide an addendum to H\aa stad's result.

\begin{theorem}\label{rank_nphard_thm} 
Tensor rank is NP-hard over fields $\mathbb{F} \supseteq \mathbb{Q}$; in particular, over $\mathbb{R}$, $\mathbb{C}$.
\end{theorem}

\begin{proof}
\cite{Haa} contains a recipe for encoding any given 3SAT Boolean formula in $n$ variables and $m$ clauses 
as a tensor $\mathcal{A} \in \mathbb{Q}^{(n+2m+2)\times3n\times(3n+m)}$ with the property that the 3SAT formula is satisfiable if and only if $\operatorname{rank}_{\mathbb{F}}(\mathcal{A})\leq4n+2m$.
The recipe defines $(n+2m+2)\times3n$ matrices for $i=1,\dots,n$, $j=1,\dots,m$:
\begin{itemize}\renewcommand{\labelitemi}{$\bullet$}
\item $V_{i}$: $1$ in $(1,2i-1)$ and $(2,2i)$, $0$ elsewhere;

\item $S_{i}$: $1$ in $(1,2n+i)$, $0$ elsewhere;

\item $M_{i}$: $1$ in $(1,2i-1)$, $(2+i,2i)$, and $(2+i,2n+i)$, $0$
elsewhere;

\item $C_{j}$: depends on $j$th clause (more involved) and has entries $0, \pm 1$;
\end{itemize}
and the $3$-tensor
\[
\mathcal{A}=[V_{1},\dots,V_{n},S_{1},\dots,S_{n},M_{1},\dots,M_{n},C_{1},\dots,C_{m}] \in  \mathbb{Q}^{(n+2m+2)\times3n\times(3n+m)}.
\]
Observe that the matrices $V_{i},S_{i},M_{i},C_{j}$ are defined with  $-1, 0, 1$ and that the argument in \cite{Haa} uses only  
the field axioms. In particular, it holds for any $\mathbb{F} \supseteq \mathbb{Q}$.
\end{proof}

\section{Symmetric tensor eigenvalue is NP-hard}\label{symm_tensor_eig}

It is natural to ask if the eigenvalue problem remains NP-hard if the general
tensor in \eqref{l2eig} or \eqref{l3eig} is replaced by a symmetric one.

\begin{problem}
[Symmetric eigenvalue]\label{prob:STEV} 
Given a symmetric tensor $\mathcal{S}\in\mathbb{Q}^{n\times n\times n}$ and $d \in \mathbb Q$, decide if $\lambda \in \mathbb{Q}(\sqrt{d})$ is an eigenvalue with \eqref{l2eig} or \eqref{l3eig}  for some $\mathbf{0} \neq \mathbf{x}\in\mathbb{R}^{n}$.
\end{problem}

As will become clear later, inputs $\lambda$ in this problem may take values in $\mathbb{Q}(\sqrt{d}) = \{a+b\sqrt{d}: a,b \in \mathbb Q\}$ for any particular $d \in \mathbb Q$.  This is not a problem since such numbers can be represented by rationals $(a,b,d)$ and arithmetic in $\mathbb{Q}(\sqrt{d})$  is rational arithmetic.

Let $G=(V,E)$ be a simple graph with vertices $V=\{1,\dots,v\}$ and edges $E$. 
A subset of vertices $S\subseteq V$ is said to be \textit{stable} (or \textit{independent}) if
$\{i,j\}\notin E$ for all $i,j\in S$, and the \textit{stability number}
$\alpha(G)$ is defined to be the size of a largest stable set. This quantity is closely related to the clique number that we encountered in Section~\ref{sec:SVP}; namely, $\alpha(G) = \omega(\overline{G})$, where $\overline{G}$ is the dual graph of $G$.
Nesterov has used the Motzkin--Straus Theorem to prove an
analogue for the stability number \cite{nesterov03,deKlerk}\footnote{We caution the reader that the equivalent of \eqref{eqNest} in \cite{nesterov03} is missing a factor of $1/\sqrt{2}$; the mistake was reproduced in \cite{deKlerk}.}.

\begin{theorem}[Nesterov]\label{thm:nesterov}
Let $G=(V,E)$ on $v$ vertices have stability number
$\alpha(G)$. Let $n=v + \frac{v(v-1)}{2}$ and $\mathbb{S}^{n-1}=\{(\mathbf{x},\mathbf{y})\in\mathbb{R}
^{v}\times \mathbb{R}^{v(v-1)/2}:\lVert\mathbf{x}\rVert_{2}^{2}+\lVert\mathbf{y}\rVert_{2}^{2}=1\}$.  Then,
\begin{equation}
\sqrt{1-\frac{1}{\alpha(G)}}=3\sqrt{\frac{3}{2}}\cdot\max_{(\mathbf{x},\mathbf{y}
)\in\mathbb{S}^{n-1}}\sum\nolimits_{i<j,\,\{i,j\}\notin E}x_{i}x_{j}y_{ij}.
\label{eqNest}%
\end{equation}
\end{theorem}

We will deduce the NP-hardness of symmetric tensor eigenvalue from the observation that every homogeneous cubic polynomial corresponds to a symmetric $3$-tensor whose
maximum eigenvalue is the maximum on the right-hand side of \eqref{eqNest}.

For any $1\le i < j < k \le v$, let
\[
s_{ijk} =
\begin{cases}
1 & 1 \le i < j \le v,\; k = v+\varphi(i, j),\; \{ i,j\} \not\in E,\\
0 & \text{otherwise},
\end{cases}
\]
where $\varphi(i,j) = (i-1)v -i(i-1)/2 + j-i$ is a lexicographical enumeration of the $v(v-1)/2$ pairs $i < j$.  For the other cases $i < k < j, \dots, k < j < i$, we set 
\[
s_{ijk} = s_{ikj} = s_{jik} = s_{jki} = s_{kij} = s_{kji}.
\]
Also, whenever two or more indices are equal, we put $s_{ijk} =0$. This defines a symmetric tensor $\mathcal{S} = \llbracket s_{ijk} \rrbracket \in \mathbb{R}^{n \times n \times n}$ with the property that
\[
\mathcal{S}(\mathbf{z},\mathbf{z},\mathbf{z})=6  \sum\nolimits_{i<j,\,\{i,j\}\notin E}x_{i}x_{j}y_{ij},
\]
where $\mathbf{z}=(\mathbf{x},\mathbf{y})\in \mathbb{R}^{v} \times \mathbb{R}^{v(v-1)/2} = \mathbb{R}^n$.

Since $\lambda=\max_{\lVert
\mathbf{z}\rVert_{2}=1}\mathcal{S}(\mathbf{z},\mathbf{z},\mathbf{z})$ is
necessarily a stationary value of $\mathcal{S}(\mathbf{z},\mathbf{z}
,\mathbf{z})$ constrained to $\lVert\mathbf{z}\rVert_{2}=1$, it is  an
$\ell^{2}$-eigenvalue of $\mathcal{S}$.  Moreover, Nesterov's Theorem implies
\[
\lambda = 2 \sqrt{\frac{2}{3}\left(1-\frac{1}{\alpha(G)}\right)}.
\]
Given a graph $G$ and $l \in \mathbb{N}$, deciding whether $\alpha(G) = l$ is equivalent to deciding whether $\omega(\overline{G}) = l$. Hence the former is an NP-complete problem given that the latter is an NP-complete problem \cite{Karp}, and we are led to the following.

\begin{theorem}\label{symm_evalue_nphard_thm}
Symmetric tensor eigenvalue over $\mathbb{R}$ is NP-hard.
\end{theorem}

\begin{proof}
For $l=v,\dots,1$, we check whether $\lambda_l = 2 \sqrt{\frac{2}{3}\left(1-\frac{1}{l}\right)}$
is an $\ell^{2}$-eigenvalue of $\mathcal{S}$. Since
$\alpha(G)\in\{1,\dots,v\}$, at most $v$ answers to Problem~\ref{prob:STEV}
with inputs $\lambda_{v},\dots,\lambda_{1}$ (taken in decreasing order of
magnitude so that the first eigenvalue identified would be the maximum) would reveal its value. Hence,  Problem~\ref{prob:STEV} is NP-hard.
\end{proof}

\begin{remark}
Here we have implicitly used the assumption that inputs to the symmetric tensor eigenvalue decision problem are allowed to be quadratic irrationalities of the form $2\sqrt{\frac{2}{3}\left(1-\frac{1}{l}\right)}$ for each integer $l \in \{1,\ldots,v\}$.
\end{remark}

It is also known that $\alpha(G)$ is NP-hard to approximate\footnote{H\aa stad's original result required $\mathit{NP} \ne \mathit{ZPP}$, but Zuckerman weakened this to $\mathit{P} \ne \mathit{NP}$.} \cite{Haa1,zuckerman2006linear} (see also the survey \cite{deKlerk}).
\begin{theorem}[H\aa stad, Zuckerman]\label{hastad_clique_thm}
It is impossible to approximate $\alpha(G)$ in polynomial time to within a factor of $v^{1-\varepsilon}$ for any $\varepsilon>0$, unless $\mathit{P}=\mathit{NP}$.
\end{theorem}


Theorem~\ref{hastad_clique_thm} implies the following inapproximability result for symmetric
tensors.

\begin{corollary}\label{cor:npzpp}
Unless $\mathit{P}=\mathit{NP}$, there is no FPTAS for approximating the largest $\ell^{2}
$-eigenvalue of a real symmetric tensor.
\end{corollary}

\section{Symmetric singular value, spectral norm, and rank-$1$ approximation are
NP-hard}\label{sec:symmproblems}

We will deduce from Theorem~\ref{symm_evalue_nphard_thm} and
Corollary~\ref{cor:npzpp} a series of hardness results for symmetric tensors
parallel to earlier ones for the nonsymmetric case. 

We first state Theorem~\ref{thm:banach1} in an alternative form; namely, that the
best rank-$1$ approximation of a symmetric tensor and its best symmetric-rank-$1$ approximation may be chosen to
be the same. Again, while we restrict ourselves to symmetric $3$-tensors, this result holds for symmetric tensors of arbitrary order. \fin
\begin{theorem}
[Banach]\label{thm:banach}Let $\mathcal{S}\in\mathbb{R}^{n\times n\times
n}$ be a symmetric $3$-tensor. Then,
\begin{equation}
\min_{\sigma\geq0,\;\lVert\mathbf{u}\rVert_{2}=\lVert\mathbf{v}\rVert
_{2}=\lVert\mathbf{w}\rVert_{2}=1}\lVert\mathcal{S}-\sigma\mathbf{u}%
\otimes\mathbf{v}\otimes\mathbf{w}\rVert_{F}=\min_{\lambda\geq0,\;\lVert
\mathbf{v}\rVert_{2}=1}\lVert\mathcal{S}-\lambda\mathbf{v}\otimes
\mathbf{v}\otimes\mathbf{v}\rVert_{F}.\label{eq:banach}%
\end{equation}
Furthermore, the optimal $\sigma$ and $\lambda$ may be chosen to be
equal. 
\end{theorem}
\begin{proof}
This result follows from carrying our discussion relating spectral norm, largest singular value, and best
rank-$1$ approximation (for nonsymmetric tensors) in Section~\ref{sec:rank1}
over to the case of symmetric tensors.  This gives
\[
\lambda=\max_{\lVert\mathbf{v}\rVert_{2}=1}\langle\mathcal{S},\mathbf{v}
\otimes\mathbf{v}\otimes\mathbf{v}\rangle=\lVert\mathcal{S}\rVert_{2,2,2},
\]
where $\lambda$ is the optimal solution for the right-hand side of \eqref{eq:banach}  and the last equality is by Theorem~\ref{thm:banach1}.
\end{proof}
 Theorems~\ref{thm:banach1} and \ref{thm:banach}, together with Theorem~\ref{symm_evalue_nphard_thm} and Corollary~\ref{cor:npzpp}, prove:

\begin{theorem}
\label{thm:symmresults}The following problems are all NP-hard over $\mathbb{F} = \mathbb{R}$:
\begin{enumerate}[\upshape (i)]
\item Deciding the largest $\ell^{2}$-singular value or eigenvalue of a symmetric $3$-tensor.
\item Deciding the spectral norm of a symmetric $3$-tensor.
\item Determining the best symmetric rank-$1$ approximation of a symmetric $3$-tensor.
\end{enumerate}
Furthermore, unless $\mathit{P}=\mathit{NP}$, there are no FPTAS for these problems.
\end{theorem}

\begin{proof}
Let $\mathcal{S}\in\mathbb{Q}^{n\times n\times n}$ be a symmetric $3$-tensor.
By Theorem~\ref{thm:banach}, the optimal $\sigma$ in the best rank-$1$
approximation of $\mathcal{S}$ (left-hand side of \eqref{eq:banach}) equals the
optimal $\lambda$ in the best symmetric rank-$1$ approximation of
$\mathcal{S}$ (right-hand side of \eqref{eq:banach}). Since the optimal $\sigma$ is
also the largest $\ell^{2}$-singular value of $\mathcal{S}$ and the optimal
$\lambda$ is the largest $\ell^{2}$-eigenvalue of $\mathcal{S}$, these also coincide.
Note that the optimal $\sigma$ is also equal to the spectral norm
$\lVert\mathcal{S}\rVert_{2,2,2}$. The NP-hardness and non-existence of FPTAS
of problems (i)--(iii) now follow from Theorem~\ref{symm_evalue_nphard_thm} and Corollary~\ref{cor:npzpp}.
\end{proof}

Theorem~\ref{thm:symmresults} answers a question in \cite{BV} about the computational complexity of spectral norm for symmetric tensors. 

\section{Tensor nonnegative definiteness is NP-hard}

There are two senses in which a symmetric $4$-tensor $\mathcal{S}\in\mathbb{R}^{n\times n\times n\times n}$ can be nonnegative definite.
We shall reserve the term \textit{nonnegative definite} to describe $\mathcal{S}$ for which $\mathcal{S}
(\mathbf{x},\mathbf{x},\mathbf{x},\mathbf{x})$ is a nonnegative polynomial;
i.e.,
\begin{equation}
\mathcal{S}(\mathbf{x},\mathbf{x},\mathbf{x},\mathbf{x})=\sum
\nolimits_{i,j,k,l=1}^{n}s_{ijkl}x_{i}x_{j}x_{k}x_{l}\geq0, \quad\text{for all $\mathbf{x}\in\mathbb{R}^{n}$.}\label{eq:psd}
\end{equation}
On the other hand, we say that $\mathcal{S}$ is
\textit{Gramian} if it can be decomposed as a positive combination of rank-$1$ terms:
\begin{equation}\label{eq:gram1}
\mathcal{S} = \sum_{i=1}^r \lambda_i \mathbf{v}_i\otimes\mathbf{v}_i\otimes\mathbf{v}_i\otimes\mathbf{v}_i, \quad \lambda_i > 0, \; \lVert \mathbf{v}_i \rVert_2 = 1;
\end{equation}
or equivalently, if $\mathcal{S}(\mathbf{x},\mathbf{x},\mathbf{x},\mathbf{x})$
can be written as a sum of fourth powers of linear forms:
\begin{equation}
\mathcal{S}(\mathbf{x},\mathbf{x},\mathbf{x},\mathbf{x})=\sum_{i=1}
^{r}(\mathbf{w}_{i}^{\top}\mathbf{x})^{4}. \label{eq:gram2}
\end{equation}
The correspondence between \eqref{eq:gram1} and \eqref{eq:gram2} is to set
$\mathbf{w}_i = \lambda_i^{1/4} \mathbf{v}_i$.
Note that for a symmetric matrix $S\in\mathbb{R}^{n\times n}$, the condition
$\mathbf{x}^{\top}S\mathbf{x}\geq0$ for all $\mathbf{x}\in\mathbb{R}^{n}$ and
the condition $S=B^{\top}B$ for some matrix $B$ (i.e.,~$S$ is a Gram matrix)
are equivalent characterizations of the positive semidefiniteness of $S$. For
tensors of even order $d>2$, condition \eqref{eq:gram1} is strictly stronger than
\eqref{eq:psd}, but both are valid generalizations of the notion of nonnegative definiteness. In fact, the cone of nonnegative definite tensors as defined by \eqref{eq:psd} and the cone of Gramian tensors as defined by \eqref{eq:gram1} are dual \cite{Rez}.

We consider tensors of order $4$  because any symmetric $3$-tensor $\mathcal{S}
\in\mathbb{R}^{n\times n\times n}$ is \textit{indefinite} since 
$\mathcal{S}(\mathbf{x},\mathbf{x},\mathbf{x})$ can take both positive
and negative values (as $\mathcal{S}(-\mathbf{x},-\mathbf{x},-\mathbf{x})=-\mathcal{S}(\mathbf{x}
,\mathbf{x},\mathbf{x})$). Order-$3$ symmetric tensors are also trivially Gramian since $-\lambda \mathbf{v} \otimes \mathbf{v} \otimes \mathbf{v} =\lambda (-\mathbf{v}) \otimes (- \mathbf{v}) \otimes (-\mathbf{v})$, and thus $\lambda$ may always be chosen to be positive.

We  deduce the NP-hardness of both notions of nonnegative definiteness  from \cite{MK} and \cite{DG}. Let $A\in\mathbb{R}^{n \times n}$ be symmetric. The matrix $A$ is said to be \textit{copositive} if $A(\mathbf{y},\mathbf{y})=\mathbf{y}^{\top}A\mathbf{y}\geq0$ for all $\mathbf{y}\geq\mathbf{0}$, and $A$ is said to be \textit{completely positive} if $A = BB^\top$ for some $B \in \mathbb{R}^{n \times r}$ with $B \ge 0$ (i.e., all entries nonnegative). The set of copositive matrices in $\mathbb{R}^{n \times n}$ is easily seen to be a cone and it is dual to the set of completely positive matrices, which is also a cone. Duality here means that $\operatorname{tr}(A_1 A_2) \ge 0$ for any copositive $A_1$ and completely positive $A_2$.

\begin{theorem}[Murty--Kabadi, Dickinson--Gijben]Deciding copositivity and complete positivity are both NP-hard.
\end{theorem}

Let $A=[a_{ij}]\in\mathbb{R}^{n \times n}$ be a symmetric matrix. Consider the symmetric
$4$-tensor $\mathcal{S}=\llbracket s_{ijkl} \rrbracket\in\mathbb{R}^{n\times n \times n \times n}$ defined by
\[
s_{ijkl}=%
\begin{cases}
a_{ij} & \text{if }i=k\text{ and }j=l\text{,}\\
0 & \text{otherwise.}%
\end{cases}
\]
Note that $\mathcal{S}$ is symmetric since $a_{ij} = a_{ji}$. Now $\mathcal{S}$ is nonnegative definite if and only if
\[
\mathcal{S}(\mathbf{x},\mathbf{x},\mathbf{x},\mathbf{x})=\sum_{i,j,k,l=1}^{n}s_{ijkl}x_{i}x_{j}x_{k}x_{l}=\sum_{i,j=1}%
^{n}a_{ij}x_{i}^{2}x_{j}^{2}\geq0,
\]
for all $\mathbf{x}\in\mathbb{R}^{n}$, which is in turn true if and only if $\mathbf{y}^{\top}A\mathbf{y}\geq0$
for all $\mathbf{y}\geq\mathbf{0}$ ($y_{i}=x_{i}^{2}$), i.e., $A$ is copositive. On the other hand, the tensor $\mathcal{S}$ is Grammian if and only if
\[
s_{ijkl} = \sum_{p=1}^r w_{ip} w_{jp} w_{kp} w_{lp}, \quad i,j,k,l=1,\dots,n,
\]
for some $r \in \mathbb{N}$, which is to say that
\[
a_{ij} = \sum_{p=1}^r w_{ip}^2 w_{jp}^2 , \quad i,j=1,\dots,n,
\]
or $A = B B^\top$ where $B = [w_{ip}^2] \in \mathbb{R}^{n \times r}$ has nonnegative entries; i.e., $A$ is completely positive. Hence we have deduced the following.
\begin{theorem}\label{thm:nonneg}
Deciding whether a symmetric $4$-tensor is nonnegative definite is NP-hard. Deciding whether a symmetric $4$-tensor is Grammian is also NP-hard.
\end{theorem}
The first statement in Theorem \ref{thm:nonneg} has appeared before in various contexts, most notably as the problem of deciding the nonnegativity of a quartic; see, for example, \cite{AOPT}. It follows from the second statement and  \eqref{eq:gram2} that deciding whether a quartic polynomial is a sum of fourth powers of linear forms is NP-hard.

The reader may wonder about a third common characterization of nonnegative definiteness: A symmetric matrix is nonnegative definite if and only if all its eigenvalues are nonnegative. It turns out that for tensors this does not yield a different characterization of nonnegative definiteness. The exact same equivalence is true for symmetric tensors with our definition of eigenvalues in Section~\ref{sec:EVP} \cite[Theorem 5]{Qi}:
\begin{theorem}[Qi]
The following are equivalent for a symmetric $\mathcal{S}\in\mathbb{R}^{n\times n\times n\times n}$: 

\begin{enumerate}[\upshape (i)]
\item $\mathcal{S}$ is nonnegative definite.

\item All $\ell^{2}$-eigenvalues of $\mathcal{S}$ are nonnegative.

\item All $\ell^{4}$-eigenvalues of $\mathcal{S}$ are nonnegative.
\end{enumerate}
\end{theorem}

The result derives from the fact that $\ell^{2}$- and $\ell^{4}$-eigenvalues are 
Lagrange multipliers. With this observation, the following is an immediate corollary of
Theorem~\ref{thm:nonneg}.

\begin{corollary}
Determining the signature, i.e., the signs of the real eigenvalues,
of symmetric $4$-tensors is NP-hard.
\end{corollary}

\section{Bivariate matrix functions are undecidable}\label{sec:biv}

While we have focused almost exclusively on \textit{complexity} in this article, we would like to add a word more about \textit{computability} in this antepenultimate section.

There has been much interest in computing various functions of a matrix \cite{higham2}. The best-known example is probably the matrix exponential $\exp : \mathbb{C}^{n\times n} \to \mathbb{C}^{n\times n}$, which is important in many applications. Recenty, there have been attempts to generalize such studies to functions of two matrices, notably \cite{Kressner}, whose approach we shall adopt. For bivariate polynomials, $f(x,y) = \sum_{i,j=0}^d a_{ij}x^i y^j \in \mathbb{C}[x,y]$, and a pair of \textit{commuting} matrices $A_1, A_2 \in \mathbb{C}^{n \times n}$, we define $f(A_1, A_2)$ as the matrix function $f(A_1,A_2) : \mathbb{C}^{n \times n} \to \mathbb{C}^{n \times n} $ given by:
\begin{equation}\label{eq:biv1}
f(A_1,A_2) (X) := \sum_{i,j=0}^d a_{ij} A_1^i X A_2^j, \ \ \ X \in \mathbb{C}^{n \times n}.
\end{equation}
Note that a pair of matrices may be regarded as a $3$-tensor $\mathcal{A} = [A_1, A_2] \in \mathbb{C}^{n \times n \times 2}$, where  $A_1, A_2$ are the two ``slices" of $\mathcal{A}$. 

If, however, we do not assume that $A_1, A_2$ be a commuting pair (these are rare, the set of commuting pairs has measure zero in $ \mathbb{C}^{n \times n \times 2}$), then \eqref{eq:biv1} is inadequate and we need to include all possible noncommutative monomials. Consider the simplest case in which $f(A_1,A_2)$ consists of a single monic monomial and $X = I$, the identity matrix, but we  no longer assume that $A_1, A_2$ necessarily commute. Then,
\begin{equation}\label{eq:biv2}
f(A_1,A_2) (I) = A_1^{m_1} A_2^{n_1} A_1^{m_2} A_2^{n_2} \cdots A_1^{m_r} A_2^{n_r},
\end{equation}
where $m_1,\dots, m_r$, $n_1,\dots, n_r$, and $r$ are nonnegative integers. If $A_1$ and $A_2$ commute, and if we write $m = m_1 + \dots + m_r$ and $n = n_1 + \dots + n_r$, then \eqref{eq:biv2} reduces to \eqref{eq:biv1} with $a_{mn} =1$ and all other $a_{ij} = 0$; i.e., $f(A_1, A_2)(I) = A_1^m A_2^n$.
Consider the following seemingly innocuous problem concerning \eqref{eq:biv2}.
\begin{problem}[Bivariate matrix monomials]\label{prob:biv}
Given $\mathcal{A} = [A_1,A_2] \in \mathbb{C}^{n \times n \times 2}$, is there a bivariate monic monomial function $f$ such that $f(A_1,A_2)(I) = 0$? 
\end{problem}
This is in fact the matrix mortality problem for two matrices \cite{HHH}, and as a consequence, we have the following.
\begin{proposition}[Halava--Harju--Hirvensalo]\label{prop:biv}
Problem~\ref{prob:biv} is undecidable when $n > 20$.
\end{proposition}
This fact is to be contrasted with its univariate equivalent: Given $A \in \mathbb{C}^{n \times n} $, a monic monomial $f$ exists with  $f(A) = 0$ if and only if $A$ is nilpotent.

\section{Open problems}\label{sec:openprobs}

We have tried to be thorough in our list of tensor problems, but there are some that we have not studied. We state a few of them here as open problems.  The first involve the hyperdeterminant.  Let $\mathbb{Q}[\mathrm{i}]:=\{a + b\mathrm{i} \in \mathbb{C} : a,b \in \mathbb{Q} \} $ be the field of \textit{Gaussian rationals}. 
\begin{conjecture}\label{hyp_approx_conjec}
Let $l,m,n \in \mathbb{N}$ satisfy  GKZ condition \eqref{eq:gkz}:
\[
l \le m+n - 1, \quad m \le l + n - 1, \quad n \le l + m - 1,
\]
and let $\operatorname{Det}_{l,m,n}$ be the $l \times m \times n$ hyperdeterminant. 
\begin{enumerate}[\upshape (i)]
\item\label{item:det0} Deciding $\operatorname{Det}_{l,m,n}(\mathcal{A}) =0$ is an NP-hard decision problem for   $\mathcal{A} \in \mathbb{Q}[\mathrm{i}]^{l \times m \times n}$.
\item\label{item:cond} It is NP-hard to decide or approximate the value for inputs $\mathcal{A} \in \mathbb{Q}[\mathrm{i}]^{l \times m \times n}$ of:
\begin{equation}\label{eq:cond}
\min_{\operatorname{Det}_{l,m,n}(\mathcal{X}) =0} \lVert \mathcal{A} - \mathcal{X}\rVert_{2,2,2}.
\end{equation}
\item\label{item:mag} Evaluating the magnitude of $\operatorname{Det}_{l,m,n}(\mathcal{A})$ is \#P-hard for inputs $\mathcal{A} \in \{ 0,1\}^{l \times m \times n}$.
\item\label{item:circuit} The homogeneous polynomial $\operatorname{Det}_{l,m,n}$ is VNP-hard to compute.
\item\label{item:cube} All statements above remain true in the special case $l = m =n$.
\end{enumerate}
\end{conjecture}

We remark that resolutions to these conjectures are likely to have implications for applications. For instance, in quantum computing, the magnitude of the hyperdeterminant in \eqref{item:mag} is the \textit{concurrence}, a measure of the amount of entanglement in a quantum system \cite{HillWoo}, and the hyperdeterminant in question would be one satisfying \eqref{item:cube}. The  decision problem in (\ref{item:det0}) is also key to deciding whether a system of multilinear equations has a nontrivial solution, as we have seen from Section~\ref{bilineareqs}.

The optimization problem \eqref{eq:cond} in (\ref{item:cond}) defines a notion of \textit{condition number} for $3$-tensors. Note that for a  non-singular  matrix $A \in \mathbb{C}^{n \times n}$, the corresponding problem for (\ref{item:cond}) has solution given by  its inverse $X = A^{-1}$ \cite[Theorem~6.5]{higham1}:
\[
\min_{\det(X) = 0} \lVert A - X \rVert_{2,2} =\lVert A^{-1} \rVert_{2,2}^{-1}.
\]
In this case, the optimum value normalized by the  spectral norm of the input gives the reciprocal of the condition number: 
\[
\frac{\lVert A^{-1}\rVert_{2,2}^{-1}}{\lVert A\rVert_{2,2}} = \kappa (A)^{-1}.
\]
Thus, for a nonzero $\mathcal{A} \in \mathbb{C}^{l \times m \times n}$,  we expect the spectral norm $\lVert \mathcal{A} \rVert_{2,2,2}$  divided by the optimum value of  \eqref{item:cond} to yield an analogue of  condition number for the tensor. Conjecture~\ref{hyp_approx_conjec} is then that the  condition number of a tensor is NP-hard to compute.


One reason for our belief in the intractability of problems involving the hyperdeterminant is that checking whether the general multivariate resultant vanishes for a system of $n$ polynomials in $n$ variables is known to be NP-hard over any field \cite{GrenetKoiranPortier}. Theorem~\ref{zerosingvaluethm}  strengthens this result  by saying that these polynomials may be chosen to be bilinear forms. Conjecture~\ref{hyp_approx_conjec}(i) further specializes by stating that these  forms \eqref{GKZbilinear} can be associated with a $3$-tensor satisfying GKZ condition \eqref{eq:gkz}. 

Another motivation for our conjectures is that the hyperdeterminant is a complex object; for instance, the $2\times 2 \times 2 \times 2$-hyperdeterminant has  2.9 million monomials \cite{huggins2008}. Of course, this does not  force the intractability of the problems above.  For instance,  the determinant and permanent of an $n \times n$ matrix have $n!$ terms, but one is efficiently computable while the other is \#P-complete \cite{Valiant79}.

In Section~\ref{sec:Rank}, we explained that tensor rank is NP-hard over any extension field $\mathbb{F}$ of $\mathbb{Q}$,
but we did not investigate the corresponding questions for the symmetric rank of a symmetric tensor (Definition~\ref{srank}). We conjecture the following.
\begin{conjecture}\label{sym_rank_conj}
Let $\mathbb{F}$ be an extension field of $\mathbb{Q}$. Let $\mathcal{S} \in \mathbb{Q}^{n \times n \times n}$  be a symmetric $3$-tensor   and $r \in \mathbb{N}$. Deciding if  $\operatorname{srank}_{\mathbb F} (\mathcal{S}) \le r$  is NP-hard.
\end{conjecture}

While tensor rank is NP-hard over $\mathbb{Q}$, we suspect that it is also undecidable.

\begin{conjecture}\label{rankQ_conj}
Tensor and symmetric tensor rank over $\mathbb{Q}$ are undecidable.
\end{conjecture}
%
%
%

We have shown that deciding the existence of an exact solution to a system
of bilinear equations \eqref{tquadeq} is NP-hard. There are two closely related
problems: (i)\ when the equalities in \eqref{tquadeq} are replaced by
inequalities and (ii) when we seek an approximate least-squares solution to
\eqref{tquadeq}.
These lead to multilinear variants of linear programming and linear least
squares. We state them formally here.

\begin{conjecture}[Bilinear programming feasibility]\label{prob:BP}
Let $A_{k},B_{k},C_{k}\in\mathbb{Q}^{n\times
n}$ and $\alpha_{k},\beta_{k},\gamma_{k}\in\mathbb{Q}$ for each $k=1,\dots,n$.
It is NP-hard to decide if  inequalities:
\begin{equation}%
\begin{cases}
\mathbf{y}^{\top}A_{i}\mathbf{z}\leq\alpha_{k}, & i=1,\dots,n;\\
\mathbf{x}^{\top}B_{j}\mathbf{z}\leq\beta_{k}, & j=1,\dots,n;\\
\mathbf{x}^{\top}C_{k}\mathbf{y}\leq\gamma_{k}, & k=1,\dots,n;
\end{cases}
\label{eq:BP}%
\end{equation}
define a nonempty subset of $\mathbb{R}^{n}$ (resp. $\mathbb{C}^{n}$).
\end{conjecture}

\begin{conjecture}[Bilinear least squares]\label{prob:BLS}
Given $3n$ coefficient matrices
$A_{k},B_{k},C_{k}\in\mathbb{Q}^{n\times n}$ and $\alpha_{k},\beta_{k}
,\gamma_{k}\in\mathbb{Q}$, $k=1,\dots,n$, the bilinear least squares problem: 
\begin{equation}
\min_{\mathbf{x},\mathbf{y},\mathbf{z}\in\mathbb{R}^{n}}\sum\nolimits_{k=1}
^{n}(\mathbf{x}^{\top}A_{k}\mathbf{y}-\alpha_{k})^{2}+(\mathbf{y}^{\top}
B_{k}\mathbf{z}-\beta_{k})^{2}+(\mathbf{z}^{\top}C_{k}\mathbf{x}-\gamma
_{k})^{2}\label{eq:BLS}
\end{equation}
is NP-hard to approximate. 
\end{conjecture}


Unlike the situation of \eqref{tquadeq}, where $\mathbf{x} = \mathbf{y} = \mathbf{z} = \mathbf{0}$ is considered a trivial solution, we can no longer  disregard an all-zero solution in \eqref{eq:BP} or \eqref{eq:BLS}. Consequently, the problem of deciding whether a homogeneous system of bilinear equations \eqref{tquadeq} has a \textit{nonzero} solution is not a special case of \eqref{eq:BP} or \eqref{eq:BLS}. 

\section{Conclusion}

Although this paper argues that most tensor problems are NP-hard, we should not be discouraged in our search for solutions to them.  For instance, while computations with Gr\"{o}bner bases are doubly exponential in the worst case \cite[pp.~400]{yap2000fundamental}, they nonetheless proved useful for Theorem~\ref{rankQRthm}. 
It is also important to note that NP-hardness is an asymptotic property; e.g., it applies to scenarios where tensor size $n$ goes to infinity.  Nonetheless, in many applications, $n$ is usually fixed and often  small; e.g., $n = 2: \lvert 0\rangle, \rvert 1\rangle$ (qubits, \cite{miyake1}), $n = 3: x,y,z$ (spatial coordinates, \cite{SchultzS}), $n = 4: A,C,G,T$ (DNA nucleobases, \cite{AllmanR}), etc. For example, while Theorem~\ref{thm:nonneg} gives an NP-hardness result for general $n$, the case $n = 3$ has a tractable convex formulation \cite{LS}.

Bernd Sturmfels once made the remark to us that ``All interesting problems are NP-hard.'' In light of this, we would like to view our article as evidence that most tensor problems are interesting.

\appendix
\section*{APPENDIX}
\setcounter{section}{1}

We give  here the complete details for the proof of Lemma~\ref{rankdroplemma}, which was key to proving Theorem~\ref{rankQRthm}. We used the symbolic computing software\footnote{One can use commercially available Maple, \url{http://www.maplesoft.com/products/maple}, or Mathematica, \url{http://www.wolfram.com/mathematica}; free \textsc{singular}, \url{http://www.singular.uni-kl.de}, Macaulay 2, \url{http://www.math.uiuc.edu/Macaulay2}, or Sage, \url{http://www.sagemath.org}. For numerical packages, see Bertini, \url{http://www.nd.edu/~sommese/bertini} and PHCpack, \url{http://homepages.math.uic.edu/~jan/download.html}.} \textsc{singular}, and in particular the function \textbf{lift}  to find the  polynomials $H_1,\dots,H_8$ and $G_1,\dots,G_8$ below.  Define three sets of polynomials:{\small
\begin{gather*}
\begin{aligned}
F_{1}  & :=a_{1}a_{2}a_{3}+c_{1}c_{2}c_{3}-2,\ F_{2}  :=a_{1}a_{3}b_{2}+c_{1}c_{3}d_{2},
\ F_{3}  :=a_{2}a_{3}b_{1}+c_{2}c_{3}d_{1},\\
F_{4}  & :=a_{3}b_{1}b_{2}+c_{3}d_{1}d_{2}+4,\ F_{5} :=a_{1}a_{2}b_{3}+c_{1}c_{2}d_{3},\ 
F_{6} :=a_{1}b_{2}b_{3}+c_{1}d_{2}d_{3}+4,\\
F_{7}  & :=a_{2}b_{1}b_{3}+c_{2}d_{1}d_{3}-4, \ F_{8}  :=b_{1}b_{2}b_{3}+d_{1}d_{2}d_{3}.
\end{aligned}\\
\begin{aligned}
G_{1}  & := -\tfrac{1}{8}b_{1}b_{2}b_{3}c_{2}d_{2}+\tfrac{1}{8}a_{2}b_{1}b_{3}d_{2}^{2}, \ G_{2}  := \tfrac{1}{8}b_{1}b_{2}b_{3}c_{2}^2-\tfrac{1}{8}a_{2}b_{1}b_{3}c_2 d_{2}, \ G_{3} :=-\frac{1}{2}c_2d_2, \ G_{4} :=\frac{1}{2}c_2^2,\\
G_{5}  & :=\frac{1}{8}a_3 b_1 b_2 c_2 d_2-\frac{1}{8}a_2 a_3 b_1 d_2^2, \  G_{6}  :=-\frac{1}{8}a_3b_1b_2c_2^2+\frac{1}{8}a_2a_3 b_1c_2 d_2, \ G_{7}   :=-\tfrac{1}{2}c_{1}^{2}, \\ G_{8} & :=\tfrac{1}{8}a_{2}a_{3}b_{1}c_{1}^{2}-\tfrac{1}{8}a_{1}a_{2}a_{3}c_{1}d_{1}.
\end{aligned}\\
\begin{aligned}
H_{1}  & :=0, \ H_{2}  :=-\tfrac{1}{32}b_{1}b_{2}b_{3}c_{2}d_{1}d_{3}+\tfrac{1}{32}a_{2}b_{1}b_{3}d_{1}d_{2}d_{3}, 
H_{3}  :=\tfrac{1}{32}b_{1}b_{2}b_{3}c_{1}d_{2}d_{3}-\tfrac{1}{32}a_{1}b_{2}b_{3}d_{1}d_{2}d_{3}, \\
H_{4} &:=\tfrac{1}{32}a_{1}b_{2}b_{3}c_{2}d_{1}d_{3}-\tfrac{1}{32}a_{2}b_{1}b_{3}c_{1}d_{2}d_{3}, \ H_{5}  :=-\tfrac{1}{8}b_{1}b_{2}b_{3}, \ H_{6}  :=\tfrac{1}{2}, H_{7}  :=\tfrac{1}{8}a_{1}b_{2}b_{3}-\tfrac{1}{8}c_{1}d_{2}d_{3}, \\  H_{8} & :=\tfrac{1}{8}c_{1}c_{2}d_{3}.
\end{aligned}
\end{gather*}}%
Both $g=2c_{2}^{2}-d_{2}^{2}$ and $h=c_{1}d_{2}d_{3}-2$ are polynomial combinations of $F_{1},\dots,F_{8}$:
\begin{equation}\label{gh_eqn}
g = \sum_{k=1}^{8}F_{k}G_{k} \quad \text{and} \quad h = \sum_{k=1}^{8}F_{k}H_{k}.
\end{equation}
Thus, if a rational point makes $F_1,\dots,F_8$ all zero, then both $g$ and $h$ must also vanish on it.  We remark that expressions such as \eqref{gh_eqn} are far from unique.


\begin{acks}
We thank Jan Draisma, Dorit Hochbaum, Jiawang Nie,  Pablo Parrilo, Sasha Razborov, Steve Smale, Bernd Sturmfels, Steve Vavasis, and Shuzhong Zhang for helpful comments that enhanced the quality of this work.  We also would like to thank the anonymous referees for comments that also improved this paper.
\end{acks}

\bibliographystyle{acmsmall}

\end{document}